\documentclass[sigconf, authorversion=true, nonacm=true]{aamas}

\usepackage{balance} 
\usepackage{comment}
\usepackage{pifont}
\usepackage{algorithm}
\usepackage{algorithmic}

\newif\ifconference
\conferencefalse

\ifconference
\setcopyright{ifaamas}
\acmConference[AAMAS '24]{Proc.\@ of the 23rd International Conference
on Autonomous Agents and Multiagent Systems (AAMAS 2024)}{May 6 -- 10, 2024}
{Auckland, New Zealand}{N.~Alechina, V.~Dignum, M.~Dastani, J.S.~Sichman (eds.)}
\copyrightyear{2024}
\acmYear{2024}
\acmDOI{}
\acmPrice{}
\acmISBN{}
\fi

\acmSubmissionID{122}

\title[Fairness and efficiency trade-off in two-sided matching]{Fairness and efficiency trade-off in two-sided matching}

\author{Sung-Ho Cho, Kei Kimura, Kiki Liu, Kwei-guu Liu, Zhengjie Liu, Zhaohong Sun, Kentaro Yahiro, and Makoto Yokoo
}
\affiliation{
  \institution{Kyushu University, Fukuoka, Japan}
  \city{}
  \country{}}
\email{{cho@agent.,kkimura@,kiki@agent.,liu@agent.,zhengjie@agent.,zhaohong.sun@,yahiro@agent.,yokoo@}inf.kyushu-u.ac.jp}

\begin{abstract}
The theory of two-sided matching has been extensively developed and
applied to many real-life application domains.
As the theory has been applied to increasingly
diverse types of environments, researchers and practitioners have
encountered various forms of distributional constraints.
As a mechanism can handle a more general class of constraints, we can assign
students more flexibly to colleges to increase students' welfare.
However, it turns out that there exists a trade-off between students' welfare
(efficiency) and fairness (which means no student has justified envy).
Furthermore, this trade-off becomes sharper as the class of constraints becomes
more general.
The first contribution of this paper is to clarify the boundary
on whether a strategyproof and fair mechanism can satisfy certain
efficiency properties for each class of constraints. 
Our second contribution is to establish 
a weaker fairness requirement called \emph{envy-freeness up to $k$ peers} (EF-$k$),
which is inspired by a similar concept used in the fair division of indivisible items.
EF-$k$ guarantees that each student has justified envy towards at most
$k$ students.
By varying $k$, EF-$k$ can represent different levels of fairness. 
We investigate theoretical properties associated with EF-$k$.
Furthermore, we develop two contrasting strategyproof mechanisms 
that work for general hereditary constraints, i.e.,
one mechanism can guarantee a strong efficiency requirement,
while the other can guarantee EF-$k$ for any fixed $k$. 
We evaluate the performance of these mechanisms through computer simulation.
\end{abstract}


\newcommand{\BibTeX}{\rm B\kern-.05em{\sc i\kern-.025em b}\kern-.08em\TeX}
\def\Mnatural{M$\sp{\natural}$}
\newcommand{\cvec}[1]{e_{#1}}

\newcommand{\red}[1]{\textcolor{red}{ #1}}
\newcommand{\blue}[1]{\textcolor{blue}{ #1}}

\makeatletter
\renewcommand{\ALG@name}{Mechanism}
\makeatother

\ifconference
\makeatletter
\gdef\@copyrightpermission{
	\begin{minipage}{0.3\columnwidth}
		\href{https://creativecommons.org/licenses/by/4.0/}{\includegraphics[width=0.90\textwidth]{figs/by.eps}}
	\end{minipage}\hfill
	\begin{minipage}{0.7\columnwidth}
		\href{https://creativecommons.org/licenses/by/4.0/}{This work is licensed under a Creative Commons Attribution International 4.0 License.}
	\end{minipage}
	\vspace{5pt}
}
\makeatother
\fi

\begin{document}


\pagestyle{fancy}
\fancyhead{}


\maketitle 

\section{Introduction}
\label{sec:intro}

The theory of two-sided matching has been developed
and 
has been applied to many real-life application domains
(see 
\citet{Roth:CUP:1990} for a comprehensive survey in
this literature).
It has attracted considerable attention from AI researchers%
~\cite{aziz2022stable,Haris19matching,hosseini2015manipulablity,IsmailiHZSY19,kawase2017near,Yahiro18,suzuki2022strategyproof}.
As the theory has been applied to increasingly
diverse types of environments, researchers and practitioners have
encountered various forms of distributional constraints
(see 
\citet{ABY:aaai22:survery} for 
a comprehensive survey on various distributional constraints). 
There exist three representative classes of constraints.
First,
the standard model of two-sided matching
considers
only the maximum quota of each individual college~%
\cite{Roth:CUP:1990}, which we call maximum quotas constraints.%
\footnote{Although our paper is described
in the context of a student-college matching problem, 
the obtained result is applicable to matching problems in general.}

More general classes of constraints 
are hereditary constraints
\cite{aziz:cutoff:2021,goto:17,kamakoji-concepts} and 
hereditary \Mnatural-convex set constraints \cite{kty:2018}.
An \Mnatural-convex set is a 
discrete counterpart of a convex set in a continuous domain. 
Hereditary constraints 
require 
that if a matching between students and colleges 
is feasible, then any matching that places
weakly fewer students at each college is also feasible.

As a mechanism can handle a more general class of constraints, we can incorporate more complex constraints required for real-life application domains. 
Also, we obtain more flexibility in
assigning students to colleges.
As a result, we can expect that students' welfare can be increased in
the obtained matching. 
Furthermore, maximum quotas constraints can be considered to be too \emph{restrictive}. In a real-life situation, it is common that some flexibility exists in determining the capacity of each college, i.e., we can increase the maximum quota of a college if it turns out to be very popular (say, by assigning additional resources). 
Such flexibility can be modeled naturally using a more general class of constraints. 

    In this paper, we focus our attention on strategyproof mechanisms, which guarantee that
    students have no incentive to misreport their preference over colleges.
    From a theoretical standpoint, if we are interested in a property achieved in
    dominant strategies,
    strategyproof mechanisms can be exclusively considered 
    without any loss of generality, as supported by the well-known revelation principle~\cite{gibbard:1973}.
    This principle states that if a certain property is satisfied in a dominant
    strategy equilibrium using a mechanism,
    it can also be achieved through a strategyproof mechanism.
    Strategyproof mechanisms are not only theoretically significant but also practically
    beneficial, as students do not need to speculate about the actions
    of others to achieve desirable outcomes; they only need to report their preferences truthfully.

Most existing works in two-sided matching require that
the obtained matching must be fair, i.e., 
no student has justified envy.
However, just requiring fairness is not sufficient
since the matching that no student is assigned to any college
is fair; we should achieve some requirement
on students' welfare (which is referred to as \emph{efficiency} in
economics) in conjunction with fairness.
In the standard maximum quotas model,
the renowned Deferred Acceptance mechanism (DA)~\cite{Gale:AMM:1962}
can achieve an efficiency property called
\emph{nonwastefulness} in conjunction with fairness.
A matching satisfying fairness and nonwastefulness together is called \emph{stable}.

However, when some distributional constraints are imposed,
there exists a trade-off between fairness and efficiency/students'
welfare. 
In particular, 
\citet{cho:2022} show that no strategyproof mechanism satisfies
fairness and a weaker efficiency property called
weak nonwastefulness under hereditary constraints. 

The first goal of this paper is to clarify the tight boundaries 
on whether a strategyproof and fair mechanism can satisfy certain
efficiency properties for each class of constraints
(see Table~\ref{tbl:properties} in Section~\ref{sec:impossibility}).
In particular, we show that 
under hereditary constraints, no strategyproof mechanism can
simultaneously satisfy fairness and a very weak efficiency requirement
called \emph{no vacant college property}.

This impossibility result illustrates a dilemma:
we are expanding/generalizing the classes of constraints 
in the hope that we can improve students' welfare.
However, if we require strict fairness, we cannot guarantee
a very weak requirement of students' welfare under general hereditary
constraints. 
Given this dilemma, our next goal is to establish
a weaker fairness requirement. In this paper, we propose a novel
concept called \emph{envy-freeness up to $k$ peers} (EF-$k$).
This concept is 
inspired by a criterion called envy-freeness up to $k$ items,
which is commonly used in the fair division of indivisible items%
~\cite{budish2011combinatorial}.
EF-$k$ guarantees that each student has justified envy towards at most
$k$ students.
By varying $k$, EF-$k$ can represent different 
levels of fairness. 
On one hand, EF-$0$ is equivalent to standard fairness. On the other hand,  any matching
satisfies EF-($n-1$), where $n$ is the number of students. 
To the best of our knowledge, this paper is the first to address the
relaxed notion of fairness in two-sided, many-to-one matching. 

We show that there exists a case that no matching is nonwasteful and
EF-$k$ for any $k<n-1$, and checking whether a nonwasteful and
EF-$k$ matching exists or not is NP-complete.
Then, we develop two contrasting strategyproof mechanisms 
that work
for general hereditary constraints.
One is based on the Serial Dictatorship mechanism (SD) \cite{goto:17},
which utilizes an optimal master-list (where students are assigned
in its order) that minimize $k$ based on colleges'
preferences, 
such that the obtained matching is guaranteed to satisfy EF-$k$.
Although $k=n-1$ holds in the worst case, we experimentally
show that $k$ tends to be much smaller when colleges' preferences are
similar.
The other one is based on the Sample and Deferred Acceptance mechanism
(SDA) \cite{liu:spr2023}, which is developed
for a special case of hereditary constraints called
student-project-resource matching-allocation problem. 
This mechanism satisfies EF-$k$ for any given $0\leq k < n-1$.
We extend SDA such that the obtained matching satisfies
no vacant college property under
a mild assumption. We experimentally show that this mechanism can significantly improve students' welfare compared to a fair (EF-$0$) mechanism 
even when $k$ is very small.

\section{Model}
\label{sec:model}
A matching market under distributional constraints is given
by $I=(S, C, X, \succ_S, \succ_C, {f})$.
The meaning of each element is as follows. 
\begin{itemize}
 \item $S=\{s_1, \ldots, s_n\}$ is a finite set of students.
 Let $N$ denote $\{1, 2, \ldots, n\}.$
 \item $C=\{c_1, \ldots, c_m\}$ is a finite set of colleges. 
Let $M$ denote $\{1, 2, \ldots, m\}$.
 \item $X \subseteq S\times C$ is a finite set of contracts.
Contract $x = (s, c) \in X$ 
represents the matching between student $s$  and college $c$.
\item For any $Y \subseteq X$, 
let $Y_s:=\{(s, c) \in Y \mid c \in
C\}$ and $Y_c:=\{(s, c) \in Y \mid s \in S\}$ denote the sets of contracts in $Y$ that involve $s$ and $c$, respectively.
\item 
$\succ_S = (\succ_{s_1}, \ldots, \succ_{s_n})$ is a profile of
the students' preferences. 
For each student $s$,  $\succ_{s}$ represents the
 preference of $s$
over $X_s \cup\{(s,\emptyset)\}$, where $(s, \emptyset)$ represents an outcome such that $s$ is unmatched. We assume $\succ_s$ is strict for each $s$.
We say contract $(s,c)$ is \emph{acceptable} for $s$ if
$(s, c) \succ_s (s, \emptyset)$ holds. 
We sometimes use notations like $c \succ_s c'$ instead of $(s,c) \succ_s (s,c')$.
\item $\succ_C = (\succ_{c_1}, \ldots, \succ_{c_m})$ is a profile of
the colleges' preferences. 
For each college $c$,  $\succ_{c}$ represents the preference of $c$
over $X_c \cup\{(\emptyset,c)\}$, where $(\emptyset,c)$ represents an outcome such that $c$ is unmatched. We assume $\succ_c$ is strict for each $c$.
We say contract $(s,c)$ is \emph{acceptable} for $c$ if
$(s, c) \succ_c (\emptyset, c)$ holds.
We sometimes write $s \succ_c s'$ instead of $(s,c) \succ_c (s',c)$.
\item ${f}: {\mathbf{Z}}_+^m \rightarrow \{-\infty, 0\}$ is a function that represents distributional constraints, 
where $m$ is the number 
of colleges and ${\mathbf{Z}}_+^m$ is 
the set of vectors of $m$ non-negative integers.
For $f$, we call a family of vectors 
$F=\{\nu \in {\mathbf{Z}}_+^m \mid f(\nu)=0 \}$ 
\emph{induced vectors} of $f$. 
\end{itemize}

We assume each contract $x$ in $X_c$ is acceptable for $c$.
This is without loss of generality because if 
some contract is unacceptable for a college, we can assume
it is not included in $X$.

We say $Y\subseteq X$ is 
a \emph{matching}, 
if for each $s \in S$, 
either (i) $Y_s=\{x\}$ and $x$ is acceptable for $s$,
or (ii) $Y_s = \emptyset$ holds.

For two $m$-element vectors $\nu, \nu' \in 
\mathbf{Z}_+^{m}$,
we say $\nu \leq \nu'$ if for all $i \in M$,
$\nu_i \leq \nu'_i$ holds.
We say $\nu < \nu'$ if $\nu \leq \nu'$ and $\nu \neq \nu'$ hold.
Also, let $|\nu|$ 
denote the $L_1$ norm of $\nu$, i.e., 
$|\nu|= \sum_{i\in M} \nu_i$. 

\begin{definition}[feasibility with distributional constraints]
Let $\nu$ be a vector of $m$ non-negative
integers. We say $\nu$ 
is \emph{feasible} in $f$ if ${f}(\nu) = 0$.
For $Y \subseteq X$, let us define $\nu(Y)$ as
$(|Y_{c_1}|, |Y_{c_2}|, \ldots, |Y_{c_m}|)$.
We say $Y$ is \emph{feasible} (in $f$) if $\nu(Y)$ is feasible in ${f}$.
\end{definition}
We assume $F$ is bounded, i.e., $|F|$ is finite. 
This is without loss of generality because 
we can assume each college $c_i$ can accept at most $|X_{c_i}|$ students, i.e., 
$f(\nu)=-\infty$ holds when $\exists i \in M, \nu_i > |X_{c_i}|$. 

%
Let us first introduce a very general class of constraints called
\emph{hereditary} constraints. 
Intuitively, heredity means that 
if $Y$ is feasible in $f$, 
then any subset $Y' \subset Y$ is also feasible in $f$. 
Let $\cvec{i}$ denote an $m$-element unit vector, where its
$i$-th element is $1$ and all other elements are $0$.
Let $\cvec{0}$ denote an $m$-element zero vector
$(0, \ldots, 0)$.
\begin{definition}[heredity]
We say a family of $m$-element vectors
$F\subseteq \textbf{Z}^m_+$ is 
\emph{hereditary} if $\cvec{0} \in F$ and 
for all 
$\nu, \nu' \in \mathbf{Z}_+^{m}$, 
if $\nu > \nu'$ and $\nu \in F$, 
then $\nu' \in F$ holds. 
We say ${f}$ is \emph{hereditary}
if its induced vectors 
are hereditary. 
\end{definition}

\citet{kty:2018} show that when 
$f$ is hereditary,  and its induced vectors satisfy 
one additional condition called \emph{\Mnatural-convexity}, 
there exists a general mechanism called Generalized Deferred Acceptance
mechanism (GDA), which satisfies several desirable properties.\footnote{%
To be more precise, 
\citet{kty:2018} show that to apply their framework, it is
necessary that the family of feasible matchings forms a matroid. 
When distributional constraints are defined on $\nu(Y)$ 
rather than on contracts $Y$, 
the fact that the family of feasible contracts forms a matroid
corresponds to the fact that (i) the family of feasible vectors forms an
\Mnatural-convex set, and (ii) it is
hereditary \cite{MS:dca:1999}.}

Let us formally define an \Mnatural-convex set. 
\begin{definition}[\Mnatural-convex set]
\label{def:mnatural}
We say a family of vectors $F \subseteq \mathbf{Z}^{m}_+$ 
forms an \emph{\Mnatural-convex set}, if for all $\nu, \nu' \in F$, 
for all $i$ such that $\nu_i > \nu'_i$, 
there exists $j \in \{0\}\cup \{k \in M \mid \nu_{k} <
\nu'_{k}\}$
such that $\nu - \cvec{i} + \cvec{j} \in F$ and 
$\nu' + \cvec{i} - \cvec{j} \in F$ hold.
We say ${f}$ satisfies \Mnatural-convexity 
if its induced vectors form an \Mnatural-convex set.
\end{definition}

An \Mnatural-convex set can be considered as a
discrete counterpart of a convex set in a continuous 
domain. Intuitively, Definition~\ref{def:mnatural} means that for two feasible vectors $\nu$ and $\nu'$, there exists another feasible vector, which is one step closer starting from $\nu$ toward $\nu'$, and vice versa.
An \Mnatural-convex set has been studied
extensively in discrete convex analysis, a branch of discrete mathematics. 
Recent advances in discrete convex analysis have found many applications
in economics (see the survey paper by \citet{murota:dca:2016}). 
Note that heredity and \Mnatural-convexity are independent properties. 

\citet{kty:2018} show that various real-life distributional constraints
can be represented as a hereditary \Mnatural-convex set. 
The list of applications includes matching markets
with regional maximum quotas \cite{kamakoji-basic}, individual/regional minimum quotas \cite{fragiadakis::2012,goto:17}, 
diversity requirements in school choice \cite{ehlers::2012,kurata:jair2017}, distance constraints \cite{kty:2018},
and so on. 
However, \Mnatural-convexity can be easily violated by introducing 
some additional constraints.

Let us introduce the most basic model where only distributional constraints are colleges' maximum quotas. 
\begin{definition}[maximum quotas]
\label{def:standard}
We say a family of vectors $F \subseteq \mathbf{Z}^{m}_+$ is given as colleges' maximum quotas, 
when  for each college $c_i \in C$, its maximum quota $q_{c_i}$ is given, and 
$\nu \in F$ iff $\forall i \in M$, $\nu_i \leq q_{c_i}$ holds.
We say ${f}$ is given as colleges' maximum quotas 
if its induced vectors are given as colleges' maximum quotas. 
\end{definition}
If $f$ is given as colleges' maximum quotas, then $f$ is a hereditary \Mnatural-convex set, but not vice versa. 

With a slight abuse of notation, for two sets of contracts
$Y$ and $Y'$,
we denote $Y_s \succ_s Y'_s$ if either (i)
$Y_s = \{x\}$, $Y'_s = \{x'\}$, and $x \succ_s x'$
for some $x, x' \in X_s$,
or
(ii) $Y_s = \{x\}$ for some $x \in X_s$ that is  acceptable for $s$
and $Y'_s = \emptyset$.
Furthermore, we denote $Y_s \succeq_s Y'_s$ if  either $Y_s \succ_s Y'_s$
or $Y_s = Y'_s$.
Also, we use notations like  $x \succ_s Y_s$ or 
$c \succ_s Y_s$, where $x$ is a contract, $Y$ is a matching, and $c$ is a college.

Let us introduce several desirable properties of a matching and a
mechanism.
We say a mechanism satisfies property A if the mechanism produces a matching that satisfies property A in every possible matching market. 

First, we define fairness. 
\begin{definition}[fairness]
\label{def:fairness}
In matching $Y$, student $s$ \emph{has justified envy} toward
another student $s'$ if 
$(s,c)\in X$ is acceptable for $s$,
$c \succ_s Y_s$, $(s', c) \in Y$, and
$s \succ_c s'$ hold.
We say matching $Y$ is \emph{fair} if no student has justified envy.
\end{definition}
Fairness implies that if student $s$ is not assigned to college $c$ (although she hopes to be assigned), 
then $c$ prefers all students assigned to it over $s$. 

Next, we define a series of properties
on students' welfare (efficiency). 
\begin{definition}[Pareto efficiency]
    Matching $Y$ is Pareto dominated by another 
    matching $Y'$ if $\forall s \in S$, $Y'_s \succeq_s Y_s$, and
    $\exists s \in S$,  $Y'_s \succ_s Y_s$ hold.
	Feasible matching $Y$ is Pareto efficient if
	no other feasible matching Pareto dominates it.
\end{definition}
\ifconference
\else
In short, feasible 
matching $Y$ is Pareto efficient
if there exists no other feasible matching 
$Y'$ such that all students weakly prefer
$Y'$ over $Y$, and at least one student 
strictly prefers $Y'$ over $Y$.
\fi

\begin{definition}[nonwastefulness]
\label{def:nonwastefulness}
In matching $Y$, student $s$ \emph{claims an empty seat} of college $c$
if 
$(s,c)$ is acceptable for $s$,
$c \succ_s Y_s$,
and $(Y \setminus Y_s) \cup \{(s, c)\}$ is feasible.
We say feasible matching $Y$ is \emph{nonwasteful} if no student claims an empty
seat. 
\end{definition}
Intuitively, nonwastefulness means that we cannot improve the matching of one student without affecting other students. 

When additional distributional constraints (besides colleges' maximum quotas) 
are imposed, 
fairness and nonwastefulness become incompatible in general.
One way to address the incompatibility is weakening the requirement of 
nonwastefulness. 
\citet{aziz:cutoff:2021} introduce a weaker efficiency concept called \emph{cut-off nonwastefulness}.
\begin{definition}[cut-off nonwastefulness]
    \label{def:cut-off-nonwastefulness}
    Feasible matching $Y$ is \emph{cut-off nonwasteful} if
    student $s$ claims an empty seat of college $c$, 
    then there exists another student $s'$ such that
    $c \succ_{s'} Y_{s'}$,
    $s' \succ_c s$, 
and $(Y \setminus Y_{s'}) \cup \{(s', c)\}$ is infeasible. 
	\end{definition}
Intuitively, 
we consider the claim of student $s$ to move her to college
$c$ from her current match is not considered legitimate if by doing
so, another student $s'$ would have justified envy toward $s$. 
\citet{aziz:cutoff:2021} show that a fair and cut-off nonwasteful matching always exists 
under hereditary constraints. This result carries over to less general hereditary and 
\Mnatural-convex set constraints, as well as weaker efficiency requirements described below. 
Note that the existence of a fair and cut-off nonwasteful matching does not guarantee 
the existence of a strategyproof mechanism for obtaining it, as shown in Section~\ref{sec:impossibility}.

\citet{kamakoji-concepts} propose another weaker version of the nonwastefulness concept, which we refer to as \emph{weak nonwastefulness}.
\begin{definition}[weak nonwastefulness]
\label{def:weak-nonwastefulness}
In matching $Y$, student $s$ \emph{strongly claims an empty seat} of $c$
if 
$(s,c)$ is acceptable for $s$,
$c \succ_s Y_s$, 
and $Y \cup \{(s, c)\}$ is feasible. 
We say feasible matching $Y$ is \emph{weakly nonwasteful} if no student strongly claims an empty seat. 
\end{definition}
\ifconference
\else
Student $s$ can strongly claim an empty seat of $c$ only when $Y\cup\{(s,c)\}$, i.e., the matching obtained by adding her to college $c$ (without removing her from her current college), is feasible. 
\fi

Let us define two more weaker efficiency properties. 
\begin{definition}[no vacant college]
\label{def:no-vacant-college}
    We say feasible matching $Y$ satisfies
\emph{no vacant college} property if student $s$ claims an empty seat
of college $c$, then $Y_s\neq \emptyset$ or $Y_c\neq \emptyset$ holds. 
\end{definition}
\ifconference
\else
Intuitively, no vacant college property means that 
the claim of student $s$ to move her to college
$c$ from her current match is considered legitimate only when
$s$ is not matched to any college and no student is assigned to $c$.
\fi

\begin{definition}[no empty matching]
    \label{def:no-empty-matching}
    In matching $Y$, 
student $s$ very strongly claims 
an empty seat of college $c$, 
when $Y=\emptyset$, 
$(s,c) \in X$, $c \succ_s \emptyset$, 
and $\{(s,c)\}$ is feasible. 
Feasible matching $Y$ satisfies no empty matching property if no student very strongly claims an empty seat of any college. 
\end{definition}
Note that this series of efficiency 
properties becomes monotonically weaker in 
this order as long as distributional constraints are hereditary. 
More specifically,
Pareto efficiency implies nonwastefulness, but not vice versa, 
nonwastefulness implies cut-off nonwastefulness, but not vice versa, and so on. 
\ifconference
\else
Pareto efficiency means that we cannot 
improve the matching of a set of students 
without hurting other students, while 
nonwastefulness means that we cannot improve the matching of \emph{one student} without 
affecting other students. Thus, Pareto efficiency implies nonwastefulness. 
If $Y$ is nonwasteful, no student can claim an empty seat. 
If $Y$ is cut-off nonwasteful, a student can claim an empty seat in some cases.
Thus, cut-off nonwastefulness is 
weaker than nonwastefulness. 
Next, we show that cut-off nonwastefulness implies weak nonwastefulness by showing its contraposition. More specifically, 
we assume 
student $s$ strongly claims an empty seat of college $c$ in $Y$. 
Then, we show that 
$Y$ cannot be cut-off nonwasteful. 
The fact that $s$ strongly claims an 
empty seat of $c$ implies that 
$s$ also claims an empty seat of 
$c$ since if 
$Y\cup \{(s,c)\}$ is feasible, 
$(Y\setminus Y_s) \cup \{(s,c)\}$ is 
also feasible. 
Assume there exists another student 
$s'$, where $c \succ_{s'} Y_{s'}$ and 
$s' \succ_c s$ hold.
Then, since 
$Y\cup \{(s,c)\}$ is feasible, 
$(Y\setminus Y_{s'}) \cup \{(s',c)\}$ is 
also feasible. Thus, $Y$ cannot be 
cut-off nonwasteful. 

Next, we show that weak nonwastefulness 
implies no vacant college property by showing its contraposition. More specifically, 
no vacant college property means that 
the claim of student $s$ to move her 
from the current matching to $c$ is 
considered legitimate only when 
$s$ is not matched to any college and 
no student is assigned to $c$. 
Let us assume $Y$ 
does not satisfy 
no vacant college property, i.e., 
there exists student $s$
who claims an empty seat of $c$ 
when 
$Y_s=\emptyset$ and $Y_c=\emptyset$. 
Then, 
we show $s$
strongly claims an empty seat of $c$ in $Y$. 
Since $Y_s = \emptyset$, 
the fact that $(Y \setminus Y_s) \cup \{(s,c)\}$ 
is feasible implies that 
$Y \cup \{(s,c)\}$ is feasible. Thus, 
$s$ also strongly claims an empty seat of $c$.
Finally, 
we show that no vacant college property 
implies no empty matching property by showing its contraposition. More specifically, 
we assume student $s$ very strongly claims 
an empty seat of $c$ in matching $Y$. 
Then, we show that $Y$ does not satisfy 
no vacant college property. 
The fact that 
student $s$ very strongly claims an empty seat of $c$ implies 
$c\succ_s \emptyset$, 
$Y_s=\emptyset$, and 
$Y_c=\emptyset$ hold. Thus,
$Y$ does not satisfy no vacant college 
property. 
\fi

Next, we introduce strategyproofness. 
\begin{definition}[strategyproofness]
\label{def:strategyproofness}
We say a mechanism is \emph{strategyproof}
if no student ever has any incentive
to misreport her preference no matter what the other students report. 
More specifically, 
let $Y$ denote the matching obtained when $s$ declares her true preference $\succ_s$, 
and $Y'$ denote the matching obtained when $s$ declare something else, 
then $Y_s \succeq_s Y'_s$ holds. 
\end{definition}

Here, we consider strategic manipulations only by students. 
It is well-known that even in the most basic model of one-to-one
matching \cite{Gale:AMM:1962}, satisfying strategyproofness (as well as basic 
fairness and efficiency requirements) for both sides is impossible~\cite{roth1982economics}. 
One rationale for ignoring the college side would be that 
the preference of a college must be
presented in an objective way and cannot be skewed arbitrarily.

\section{Existing mechanism}
In this section, we briefly introduce existing mechanisms, which are strategyproof for a given class of constraints. 
First, let us introduce Generalized Deferred Acceptance mechanism (GDA), which works  
under hereditary \Mnatural-convex set constraints \cite{Hatfield:AER:2005}. 
As its name shows, it is
a generalized version 
of the Deferred Acceptance mechanism \cite{Gale:AMM:1962}.
To define GDA, we first introduce \emph{choice functions} of 
students and colleges.

\begin{definition}[students' choice function]
For each student $s$, her \emph{choice function} $Ch_s$ 
specifies her most preferred contract within each $Y \subseteq X$,
 i.e.,
$Ch_s(Y)= \{x\}$, where $x$ is the most preferred acceptable contract
in $Y_s$ 
if one exists, and $Ch_s(Y)=\emptyset$  if no such contract exists. 
Then, the choice function of all students is defined as
$Ch_S(Y) := \bigcup_{s \in S} Ch_s(Y_s)$.
\end{definition}

\begin{definition}[colleges' choice function]
\label{def:college:choice}
We assume each contract $(s,c) \in X$ is associated with 
its unique strictly positive weight $w((s,c))$. 
We assume these weights respect each college's preference
$\succ_c$, i.e., if $(s, c) \succ_c (s', c)$, 
then $w((s,c)) > w((s', c))$ holds. 
For $Y\subseteq X$, let $w(Y)$ denote $\sum_{x \in Y} w(x)$.
Then, the choice function of all colleges is defined as
$Ch_C(Y) := \arg\max_{Y'\subseteq Y} 
{f}(\nu(Y')) + w(Y')$.
\end{definition}

As long as ${f}$ induces a hereditary 
\Mnatural-convex set, a unique subset $Y'$ exists
that maximizes the above formula. 
Furthermore, such a subset can be efficiently 
computed in the following greedy way. 
Let $Y'$ denote the set of chosen contracts, which is initially $\emptyset$. 
Then, sort $Y$ in the decreasing order of their weights. 
Then, choose contract $x$ from $Y$ one by one and add it to 
$Y'$, as long as $Y'\cup\{x\}$ is feasible. 

Using $Ch_S$ and $Ch_C$, GDA is defined as Mechanism~\ref{alg:gda}.
\begin{algorithm}[t]
\begin{algorithmic}[1]
\REQUIRE $X, Ch_S, Ch_C$
\ENSURE matching $Y$
\caption{Generalized Deferred Acceptance (GDA)}
\label{alg:gda}
\STATE $Re\leftarrow \emptyset$. 
\STATE Each student $s$ offers her most preferred contract $(s,c) $ which has not been rejected before
(i.e., $(s,c) \not\in Re$). If no remaining contract is acceptable  for $s$, $s$ does not make any offer.
Let $Y$ be the set of contracts offered (i.e., $Y=Ch_S(X\setminus Re)$). 
\STATE Colleges tentatively accept $Z=Ch_C(Y)$ and reject other contracts in $Y$ (i.e., $Y\setminus Z$).
\STATE If all the contracts in $Y$ are tentatively accepted at 3, then let $Y$ be the final matching 
and terminate the mechanism. Otherwise, $Re \leftarrow Re\cup (Y\setminus Z)$, and go to $2$.
\end{algorithmic}
\end{algorithm}
Note that we describe the mechanism using terms like "student $s$ offers" to make the description 
more intuitive. In reality, GDA is a direct-revelation mechanism, where the mechanism first collects 
the preference of each student, and the mechanism chooses a contract on behalf of each student. 

\citet{kty:2018} show that when $f$ induces a 
hereditary \Mnatural-convex set, GDA is strategyproof,
the obtained matching $Y$ satisfies a property called Hatfield-Milgrom
stability (HM-stability), and $Y$ is the student-optimal matching within all HM-stable matchings
(i.e., all students weakly prefer $Y$ over any other HM-stable matching). 
\begin{definition}[HM-stability]
Matching $Y$ is HM-stable if $Y=Ch_S(Y)=Ch_C(Y)$, and 
there exists no contract $x \in X\setminus Y$, such that 
$x \in Ch_S(Y\cup \{x\})$ and 
$x \in Ch_C(Y\cup \{x\})$ hold.
\end{definition}
Intuitively, HM-stability means there exists no contract in $X\setminus Y$ that is mutually preferred by students and colleges.
Note that HM-stability implies fairness. 
If student $s$ has justified envy in matching $Y$, there exists
$(s,c) \in X\setminus Y$, $(s', c) \in Y$, s.t. 
$(s,c) \succ_s Y_s$ and $w((s,c))> w((s',c))$ holds. 
Then, $(s,c) \in Ch_S(Y\cup \{(s,c)\})$ and 
$(s,c) \in Ch_C(Y\cup \{(s,c)\})$ hold, i.e., 
$Y$ is not HM-stable. 

\ifconference
\else
For standard maximum quotas constraints, 
the only distributional constraints are $(q_c)_{c \in C}$, i.e.,  
${f}(Y) = 0$ iff for each $c\in C$, $|Y_c|\leq q_c$ holds. 
Then, $Ch_C(Y)$ is defined as 
$\bigcup_{c\in C} Ch_c(Y_c)$, where
$Ch_c$ is the choice function of each college $c$, 
which chooses top $q_c$ contracts from $Y_c$ 
based on $\succ_c$. When $Ch_C$ is defined this way, 
GDA becomes equivalent to the standard DA. 
\fi

Next, we introduce two mechanisms that work for hereditary constraints.
The Serial Dictatorship (SD) mechanism \cite{goto:17}
is parameterized by an exogenous serial order over the students
called a master-list. 
We denote the fact that $s$ is placed in a higher/earlier position than 
student $s'$ in master-list $L$ as $s \succ_{L} s'$. 
Students are assigned sequentially according to the master-list. 
In our context with constraints, student $s$ is assigned to 
her most preferred college $c$, 
where $c$ considers her acceptable (i.e., $(s,c) \in X$ holds) and 
assigning $s$ to $c$ does not cause any constraint violation. 
More specifically, assume the obtained matching for students placed higher than 
$s$ in $L$ is $Y$. Then, $s$ can be assigned to $c$ when 
$f(\nu(Y\cup\{(s,c)\}))=0$ holds. 
SD is strategyproof and achieves Pareto efficiency.

The Artificial Cap Deferred Acceptance mechanism (ACDA) is defined 
as follows. First, 
we choose one vector $\nu^*$ s.t. $f(\nu^*)=0$, and 
there exists no $\nu' > \nu^*$ where $f(\nu')=0$, 
i.e., a maximal feasible vector. Note that $\nu^*$ must be chosen independently from 
students' preferences $\succ_S$ to guarantee strategyproofness. 
Then, we apply standard DA, where 
maximum quota $q_{c_i}$ for each college $c_i$ is given as $\nu^*_i$. 
Intuitively, in ACDA, the set of feasible vectors $F$ is artificially reduced to 
a hyper-rectangle, where $\nu$ is feasible iff $\nu\leq \nu^*$. 
ACDA is strategyproof and fair, assuming $\nu^*$ is chosen independently from 
students' preferences. 

\section{Existence of fair and strategyproof mechanism}
\label{sec:impossibility}
\begin{table}
\caption{Existence of fair and strategyproof mechanism
(\ding{51} means such a mechanism exists, 
 \ding{55} means such a mechanism does not exist, and 
 \ding{54} means even without strategyproofness, 
 a matching that satisfies fairness and the efficiency property 
 may not exist. A red mark represents a new result obtained in this paper) }
 \label{tbl:properties}
    \begin{center}
          \begin{tabular}{clll}
    & maximum 
    & hereditary \&
    & hereditary \\
    & quotas  &  \Mnatural-convex & \\
    & & set & \\
   \hline 
   Pareto efficiency & \ding{54} \cite{roth1982economics}
              & \ding{54}  & \ding{54} \\
    nonwastefulness & \ding{51} \cite{roth1985} &  \ding{54} \cite{kamakoji-concepts} &  \ding{54} \\
    cut-off nonwastefulness & \ding{51} &  \red{\ding{55}} [Thm~\ref{thm:no-existence-cutoff}] & \ding{55} \\    
    weak nonwastefulness  & \ding{51} &  \ding{51} \cite{Kimura23} & \ding{55} \cite{cho:2022} \\
    no vacant college  & \ding{51} &  \ding{51} & \red{\ding{55}} [Thm~\ref{thm:no-existence-no-vacant}]
    \\
    no empty matching  & \ding{51} &  \ding{51} & \red{\ding{51}} [Thm~\ref{thm:existence-no-empty-matching}]\\
    \hline
  \end{tabular}
    \end{center}
\end{table}
In this section, we examine whether a fair and strategyproof mechanism exists under a given class of distributional constraints in conjunction with some efficiency property. 
The classes of constraints we consider are: 
maximum quotas constraints, hereditary and \Mnatural-convex set constraints, and hereditary constraints. 

First, we list known results. 
\begin{itemize}
\item   For maximum quotas constraints, fairness, nonwastefulness, and strategyproofness are compatible, i.e., the standard DA satisfies these properties
\cite{roth1985}.
On the other hand, 
fairness and Pareto efficiency are incompatible, i.e., even without strategyproofness, 
a matching that satisfies Pareto efficiency and fairness 
may not exist \cite{roth1982economics}.

\item For hereditary and \Mnatural-convex set constraints, 
fairness, weak nonwastefulness, and strategyproofness are compatible, i.e., 
Generalized DA satisfies these properties \cite{Kimura23}.
On the other hand, fairness and  nonwastefulness are incompatible \cite{kamakoji-concepts}.
\item For hereditary constraints, fairness, weak nonwastefulness, and strategyproofness are incompatible \cite{cho:2022}
\end{itemize}

Given these known results, the remaining open questions are as follows. 
\begin{enumerate}
    \item Under hereditary and \Mnatural-convex set constraints,  
    does a strategyproof, fair, and cut-off nonwasteful mechanism exist? 
    \item Under hereditary constraints, can a strategyproof 
    and fair mechanism satisfy any property weaker than 
    weak nonwastefulness? 
\end{enumerate}
For question (1), we obtain a negative answer, as shown 
in Theorem~\ref{thm:no-existence-cutoff}.
For question (2), we obtain a stronger result than
\citet{cho:2022}, i.e., 
Theorem~\ref{thm:no-existence-no-vacant} shows 
that no mechanism simultaneously satisfies 
strategyproofness, fairness, and no vacant college 
property. Then, we show a simple mechanism that 
satisfies strategyproofness, fairness, and no empty matching 
property (Theorem~\ref{thm:existence-no-empty-matching}). In summary, we obtain tight boundaries 
(at least in the granularity of efficiency properties we consider in this paper) on whether a strategyproof and fair mechanism can satisfy certain efficiency properties for each class of constraints
(Table~\ref{tbl:properties}).

\begin{theorem}
\label{thm:no-existence-cutoff}
No mechanism can simultaneously satisfy 
fairness, strategyproofness, and cut-off nonwastefulness under hereditary \Mnatural-convex
set constraints.
\end{theorem}
\ifconference
\else
\begin{table}[!t]
        \caption{Possible matchings for preference profiles\\ 
        (Theorem~\ref{thm:no-existence-cutoff})}        \label{table:cutoff-no-fair}
        \centering
        \begin{tabular}{llll}
            \hline preference & 
            $s_1$  & $s_2$ & possible  \\
            profile & & & matchings \\
            \hline \hline
            $\succ^1_S$  & $c_1$$c_2$  & $c_2$ & 
            [$c_1, \emptyset$]\\ \hline
            $\succ^2_S$  & $c_1$ & $c_2 c_1$ & [$\emptyset, c_2$] \\ \hline
            $\succ^3_S$  & $c_1$ & $c_2$ & $[c_1, \emptyset$], [$\emptyset, c_2$]\\ \hline
        \end{tabular}
\end{table}

    \begin{proof}
    Consider a matching market with two 
    students $S = \{s_1, s_2\}$ and two 
    colleges $C = \{c_1, c_2\}$.
    The colleges' preference profile 
    $\succ_C$ are as follows:
    \[
    \begin{array}{ccc}
        c_1: s_2 \succ_{c_1} s_1\\
        c_2: s_1 \succ_{c_2} s_2\\
    \end{array}
    \]

     To make the description concise, 
    we denote a preference of students by 
    a sequence of acceptable colleges. 
    For example, we denote 
    $c_1 \succ_{s} c_2 \succ_s \emptyset$ 
    as $c_1 c_2$, and 
    $c_1 \succ_{s} \emptyset \succ_{s} c_2$ 
    as $c_1$. 
    Furthermore, we denote a matching
    as a pair of colleges assigned to $s_1$ and $s_2$.
    For example, we denote matching
    $\{(s_1, c_2), (s_2, c_1)\}$ as 
    $[c_2, c_1]$. 
   
    Suppose $f(\nu) = 0$ if and only if $\nu \leq \nu'$ for some $\nu' \in \{(1, 0), (0, 1)\}$.
    This setting reflects the situation where the regional quotas constraints $|Y_{c_1}| + 
    |Y_{c_2}| \leq 1$ 
    are imposed, which form \Mnatural-convex set constraints.

    Assume, for the sake of contradiction, 
    that there exists a fair, strategyproof, and 
    cut-off nonwasteful mechanism. 
    We examine three students' preference profiles: 
    $\succ^1_S, \succ^2_S$, and $\succ^3_S$. 
    These preference profiles and possible matchings that satisfy fairness and cut-off nonwastefulness are summarized in 
    Table~\ref{table:cutoff-no-fair}. 
    First, for 
    $\succ_S^1 = (c_1 c_2, c_2)$,  due to fairness, we cannot allocate $s_2$ to $c_2$.
    Also, due to cut-off nonwastefulness, 
    we cannot allocate $s_1$ to $c_2$. 
    Then, the mechanism must choose 
    $[c_1, \emptyset]$. 

    Next, for 
    $\succ_S^2 = (c_1, c_2 c_1)$, due to fairness, we cannot allocate $s_1$ to $c_1$. 
    Also, due to cut-off nonwastefulness, 
    we cannot allocate $s_2$ to $c_1$. 
    Then, the mechanism must choose 
    $[\emptyset, c_2]$. 

    Finally, for $\succ_S^3 = (c_1, c_2)$, 
    due to cut-off nonwastefulness and distributional 
    constraints, exactly one student must be assigned 
    to her acceptable college. Thus, there exist two possible matchings:
    (a) $[c_1, \emptyset]$ or (b) $[\emptyset, c_2]$. 
   If (a) is chosen, then $s_2$ has an incentive to
   manipulate (to modify the profile to $\succ_S^2$)
   so that she is assigned to $c_2$. 
   If (b) is chosen, then $s_1$ has an incentive 
   to manipulate (to modify the profile to $\succ_S^1$)
   so that she is assigned to $c_1$. 
   This fact violates our assumption that 
   the mechanism is strategyproof. 
    \end{proof}

\fi
 
\begin{theorem}\label{thm:no-existence-no-vacant}
No mechanism can simultaneously satisfy 
fairness, strategyproofness, and no vacant college 
property under hereditary constraints. 
\end{theorem}
\ifconference
\else
\begin{table}[!t]
        \caption{Possible matchings for preference profiles\\ (Theorem~\ref{thm:no-existence-no-vacant})}        \label{table:fair_NVH4}
        \centering
        \begin{tabular}{llll}
            \hline preference & 
            $s_1$  & $s_2$ & possible  \\
            profile & & & matchings \\
            \hline \hline
            $\succ^1_S$  & $c_2$  & $c_1$ & 
            \blue{[$c_2, c_1$]}\\ \hline
            $\succ^2_S$  & $c_2$ & $c_1 c_3$ & \blue{[$c_2, c_1$]}, [$\emptyset, c_3$] \\ \hline
            $\succ^3_S$  & $c_1 c_2 c_3$ & $c_1 c_3$ & \blue{[$c_1, \emptyset$]}, [$\emptyset, c_3$]\\ \hline
            $\succ^4_S$  & $c_1 c_2 c_3$  & $c_1 c_3 c_4$ & \blue{[$c_1, \emptyset$]}, [$\emptyset, c_3$]\\ \hline
            $\succ^5_S$  & $c_3 c_1$  & $c_1 c_3 c_4$ & \blue{[$c_1, \emptyset$]}, [$\emptyset, c_3$]\\ \hline
            $\succ^6_S$  & $c_3 c_1$  & $c_4$ & \blue{[$c_1, \emptyset$]}, [$c_3, c_4$]\\ \hline
            $\succ^7_S$  & $c_3$  & $c_4$ & [$c_3, c_4$]\\ \hline
        \end{tabular}
\end{table}

 \begin{proof}
     Consider a matching market with two students $S = \{s_1, s_2\}$ and four colleges $C = \{c_1, c_2, c_3, c_4\}$.
		The colleges' preferences $\succ_C$ are as follows:
		\[
		\begin{array}{ccc}
			c_1: s_1 \succ_{c_1} s_2\\
			c_2: s_1 \succ_{c_2} s_2\\
			c_3: s_2 \succ_{c_3} s_1\\
			c_4: s_2 \succ_{c_4} s_1\\
		\end{array}
		\]
    Suppose $f(\nu) = 0$ if and only if $\nu \leq \nu'$ for some $\nu' \in \{(1, 1, 0, 0), \linebreak 
    (0, 0, 1, 1)\}$.

    Assume, for the sake of contradiction, that 
    there exists a mechanism that is fair, 
    strategyproof, and satisfies no vacant college 
    property. 

Here, we examine seven possible students' profiles
$\succ^1_S, \ldots, \succ^7_S$ described in 
Table~\ref{table:fair_NVH4}.
For each students' profile, we also enumerate 
all matchings that are fair and satisfy 
no vacant college property. 
For $\succ^1_S$, due to no vacant college property, 
both students must be assigned to their first 
choice colleges. Thus, the only possible matching 
is $[c_2, c_1]$. 
For $\succ^2_S$, another matching, 
[$\emptyset, c_3$] is also possible. 
However, if the mechanism chooses 
[$\emptyset, c_3$], student $s_2$ has an incentive 
to manipulate (to modify the profile to $\succ^1_S$)
so that she is assigned to $c_1$. 
Thus, the mechanism must choose $[c_2, c_1]$. 
For $\succ^3_S$, both students consider 
$c_1$ and $c_3$ acceptable. Due to fairness, 
only $s_1$ can be assigned to $c_1$, and 
only $s_2$ can be assigned to $c_3$. 
Also, if we assign $s_1$ to $c_2$, due to 
no vacant college property, we need to assign 
$s_2$ to $c_1$, However, this violates fairness. 
Thus, possible matchings are either 
[$c_1, \emptyset$] or [$\emptyset, c_3$]. 
However, if the mechanism chooses 
[$\emptyset, c_3$], 
student $s_1$ has an incentive to manipulate 
 (to modify the profile to $\succ^2_S$)
 so that she is assigned to $c_2$. 
 Continuing a similar argument, we obtain that the mechanism must choose the matching colored in 
 blue in Table~\ref{table:fair_NVH4}.
 In particular,  for $\succ^6_S$, the mechanism must choose 
 [$c_1, \emptyset$]. 
For $\succ^7_S$, the only matching that satisfies no vacant college 
property is [$c_3, c_4$].
This implies that when the profile is $\succ^6_S$, 
student $s_1$ has an incentive 
to manipulate (to modify the profile to $\succ^7_S$)
so that she is assigned to $c_3$. 
This violates our assumption that the mechanism is 
strategyproof.
 \end{proof}
\fi
 
Next, we show that there exists a mechanism 
that satisfies fairness, strategyproofness, and 
no empty matching property under hereditary constraints.
This mechanism utilizes GDA. 
More specifically, 
for given $f$, which is hereditary, we construct a set of vectors $F'$
such that 
$\forall \nu \in F'$, $f(\nu)=0$ holds (i.e., 
$F'$ is a subset of vectors induced by $f$), 
and $F'$ is a hereditary \Mnatural-convex set. 
Then, we apply GDA by using $f'$ (where 
$f'(\nu)=0$ iff $\nu\in F'$) instead of $f$. 
$F'$ is constructed as follows. 
We initialize $F'\leftarrow \{\cvec{0}\}$. 
Then, for each $i \in M$, if
$f(\cvec{i})=0$, we add $\cvec{i}$ to $F'$. 
Clearly, $F'$ is an \Mnatural-convex set;
it contains only $\cvec{0}$ and $\cvec{i}$ ($i\in M$).

\begin{theorem}
\label{thm:existence-no-empty-matching} 
Under hereditary constraints, GDA using 
$f'$ is fair, strategyproof, and satisfies
no empty matching property. 
\end{theorem}
\begin{proof}
For the obtained matching $Y$ by GDA, 
$f'(\nu(Y))=0$ holds. Then, by way of constructing 
$F'$, $f(\nu(Y))=0$ holds, i.e., $Y$ is feasible. 
Since $f'$ induces a hereditary \Mnatural-convex set, 
GDA is strategyproof and fair \cite{kty:2018}. 
Also, as long as 
there exists $(s,c) \in X$ such that 
$c \succ_s \emptyset$ and 
$f'(\nu(\{(s,c)\})=0$ hold, 
$Y\neq \emptyset$ holds.
This is because if $Y=\emptyset$, then 
$(s,c) \in Ch_S(Y\cup\{(s,c)\})$
and $(s,c) \in Ch_C(Y\cup\{(s,c)\})$ hold, 
which violates the fact that 
GDA obtains an HM-stable matching. 
\end{proof}

\section{New fairness concept: Envy-Free up to $k$ peers (EF-$k$)}
In this section, we introduce a weaker fairness concept
called envy-free up to $k$ peers (EF-$k$). 
For matching $Y$ and student $s$, 
let $Ev(Y, s)$ denote 
$\{s' \mid s'\in S, s \text{ has } 
\text{justified envy toward }
s' \text{ in } Y\}$.

\begin{definition}[Envy-free up to $k$ peers]
Matching $Y$ is envy-free up to $k$ peers 
(EF-$k$) if $\forall s \in S$, 
$|Ev(Y,s)|\leq k$ holds. 
\end{definition}
EF-$0$ is equivalent to fairness. 
Any matching is EF-$(n-1)$, where $n=|S|$. 

There are other ways to relax fairness than EF-$k$. One straightforward way is to minimize the total number of justified envies. However, this criterion can be 
\emph{unfair} among students, e.g., one student has many envies while others have only a few. Our definition of EF-$k$ is more egalitarian; it minimizes the envies of the worst student. Other egalitarian criteria are also possible. For example, instead of counting the number of students to whom each student has envy, we can count the colleges at which each student has envy. Also, we can count the number of students by whom each student is envied. Which concept is socially acceptable is difficult to tell. This work is a first step that brings up new research directions in two-sided matching, i.e., how to relax the fairness concept in a socially acceptable way.

We use the following example to show that 
nonwastefulness and EF-$k$ are incompatible 
for any $k < n-1$
under hereditary \Mnatural-convex set constraints. 
\begin{example}
\label{ex:regional}
There are $n$ students and $n$ colleges. 
For each student $s_i$, her preference is:
$c_{i+1} \succ_{s_i} c_{i+2} \succ_{s_i} 
\ldots 
\succ_{s_i} c_n \succ_{s_i} c_1 
\succ_{s_i} \ldots \succ_{s_i} c_{i}$.
For each college $c_i$, its preference is:
$s_i \succ_{c_i} s_{i+1} 
\succ_{c_i} \ldots
\succ_{c_i} s_{n} \succ_{c_i} s_1 
\succ_{c_i} \ldots \succ_{c_i} s_{i-1}$. 
In short, for each student $s_i$, her most preferred 
college $c_{i+1}$ considers her as the least preferred student, 
and her least preferred college $c_i$ considers 
her as the most preferred student. 
Distributional constraints $f$ is defined as:
$f(\nu)=0$ iff $\forall i \in M$, 
$|\nu_i| \leq 1$ and 
$\sum_{i \in M} |\nu_i| \leq n-1$ hold, i.e., each college can accept at most one student, and the total number of students accepted to all colleges 
is at most $n-1$. Clearly, $f$ induces a hereditary 
\Mnatural-convex set. 
\end{example}

\begin{theorem}
\label{thm:no-existence-EF-k}
Under hereditary \Mnatural-convex set constraints, there exists a case that no matching is nonwasteful and 
EF-$k$ for any $k< n-1$.
\end{theorem}
\begin{proof}
Consider the setting in Example~\ref{ex:regional}. 
The total number of accepted students is 
at most $n-1$. Also, due to nonwastefulness, 
exactly one student is unassigned to any college. 
By symmetry, without loss of generality, 
let us assume $s_1$ is unassigned. 
Then, there exists exactly one vacant college, i.e., a college to which no student is assigned. 
The vacant college must be $c_2$, since if $c_i$ ($i\neq 2$) 
is vacant, student $s_{i-1}$ claims an empty seat 
of $c_i$. 
Also, $s_n$ must be assigned to $c_1$. Otherwise, 
she is assigned to $c_i$ where $3 \leq i \leq n$; she claims an empty seat of $c_2$. 
Then, $s_{n-1}$ must be assigned to $c_n$. 
Otherwise, she is assigned to $c_i$ where $3 \leq i \leq n-1$; she claims an empty seat of $c_2$. 
By repeating a similar argument, we obtain that
each student $s_i$ ($i\neq 1$) is assigned to 
her most preferred college $c_{i+1}$. 
Then, $s_1$ has justified envy toward $s_2, \ldots, s_n$.
Thus, $|Ev(Y,s_1)|=n-1$ holds. 
\end{proof}

Given Theorem~\ref{thm:no-existence-EF-k}, 
a natural question is the complexity of 
checking the existence of 
a nonwasteful and EF-$k$ matching (for $k<n-1$). 
Let us assume $f$ can be computed in a constant time. 
\ifconference
\else
To examine this complexity, 
we utilize the following lemma.
\begin{lemma}
\label{lem:cut-off} 
Checking whether a fair and nonwasteful matching exists or not is NP-complete, even when distributional constraints form a hereditary \Mnatural-convex set. 
\end{lemma}
\begin{proof}
\citet{aziz:cutoff:2021} show that checking the 
existence of a 
\emph{strongly stable} matching is NP-complete 
for \emph{REG} constraints. 
Strong stability is equivalent to fairness and 
nonwastefulness. \emph{REG} constraints mean 
regional maximum quotas for mutually disjoint
regions, which is a special 
case of hereditary \Mnatural-convex set constraints. 
Thus, this complexity result carries over to 
hereditary \Mnatural-convex constraints, which is 
more general than \emph{REG}. 
\end{proof}
\fi

\begin{theorem}
\label{thm:EF-k-complexity}
Checking whether an EF-$k$ ($k < n-1$)  and nonwasteful matching exists or not is NP-complete, even when distributional constraints form a hereditary \Mnatural-convex set.
\end{theorem}
\ifconference
\else
\begin{proof}
First, for given matching $Y$, we can check whether $Y$
is EF-$k$ and nonwasteful in polynomial time, so the problem is 
in NP. Next, we show a reduction from the problem of 
checking whether a fair and nonwasteful matching exists or not. 
 Consider an original matching problem instance $I$, where distributional constraints form a hereditary 
 \Mnatural-convex set. 
 We create an instance of an 
 extended market $I'$
 as follows.
 \begin{itemize}
     \item 
 For each college in $I$, we create 
 a corresponding college $c'$ in $I'$. 
 Let $C'$ denote the set of these colleges
 in $I'$. The distributional constraints over
 $C'$ are the same as the original instance $I$. 
\item  For each student $s_i$ in $I$, we create
 $k+1$ students $s_{i,1}, \ldots, s_{i,k+1}$, 
 as well as $k+1$ additional colleges
 $c_{i,1}, \ldots, c_{i,k+1}$.
These additional colleges for $s_i$ 
form a region with regional maximum quota $k$.
Each student in $s_{i,1}, \ldots, s_{i,k+1}$
is a copy of student $s_i$ in the original instance $I$.
\item The preference of additional college
$c_{i,j}$ is defined in the same way as 
Example~\ref{ex:regional}.
More specifically, 
each additional college $c_{i,j}$ can accept only corresponding 
(copied) students $s_{i,1}, \ldots, s_{i,k+1}$, 
and $c_{i,j}$ most prefers $s_{i,j}$ and least 
prefers $s_{i,j-1}$.
\item Each student $s_{i,j}$ prefers any of its 
additional colleges over any original college.
The order of original colleges is the same as 
the original instance $I$. 
The order of her additional colleges is 
defined in the same way as Example~\ref{ex:regional}, i.e., 
$s_{i,j}$ most prefers $c_{i,j+1}$.
\item The preference of each college $c' \in C'$
is defined as follows. 
If $s_i \succ_c s_j$ holds in the original instance, 
$s_{i, t} \succ_{c'} s_{j, t'}$ holds 
for any $t, t' \in \{1, \ldots, k+1\}$. 
The preference over $s_{i,1}, \ldots, s_{i,k+1}$, i.e., 
the copied students of the same original 
student, can be decided arbitrarily. 
 \end{itemize}  
We can observe the following facts. 
Matching $Y$ in the extended instance $I'$ 
is nonwasteful only when for each $i\in N$ 
and copied students $s_{i,1}, \ldots, s_{i,k+1}$, 
exactly $k$ students are assigned to their 
additional colleges $c_{i,1}, \ldots, c_{i, k+1}$. 
Also, these $k$ students must be assigned to 
their first-choice colleges. Thus, the only student
who is not assigned to her additional colleges
has justified envy toward other $k$ copied students. 
Let $S'$ denote the set of students who are not assigned to their additional colleges.
$S'$ will be assigned to $C'$. 
Assume $Y$ is EF-$k$ and 
nonwasteful, then the matching between 
$S'$ and $C'$ within $Y$ must be 
nonwasteful and fair; otherwise, at least one student in $S'$ has justified envy toward more than $k$ students or 
$Y$ is wasteful for the original instance 
(to obtain a matching in the original instance
from $Y$, 
we replacing $s_{i,j}$ to $s_i$ and $c'$ to $c$). 
Also, if there exists a fair and nonwasteful matching in the original instance $I$,  then there exists an EF-$k$ and nonwasteful matching in 
$I'$; the assignment of $s_{i,1}$ is the same as
$s_i$, and the rest of the students are assigned to their favorite additional colleges.  
\end{proof}
\fi

\section{New mechanisms}
In this section, we introduce 
two contrasting strategyproof mechanisms that work
for general hereditary constraints.
The first one (called SD$^*$) satisfies the strongest efficiency property, i.e.,
Pareto efficiency, while it cannot guarantee EF-$k$ for any fixed
$k < n-1$. The second one (called SDA with reserved quotas)
satisfies EF-$k$ for any fixed $k< n-1$,
while it can only guarantee a rather weak efficiency property.
In the next section, we experimentally show that
SD$^*$ can guarantee EF-$k$ where $k$ is much smaller than $n-1$ when
colleges' preferences are similar. 
Furthermore, 
we experimentally show that SDA with reserved quotas
can significantly improve students' welfare compared to a fair (EF-$0$) mechanism 
even when $k$ is very small. 

\subsection{Pareto efficient mechanism}
First, we develop a strategyproof and Pareto efficient mechanism 
based on SD.
For master-list $L$, a pair of students $(s, s')$, and college $c$, 
we say $c$ disagrees with $L$ for $(s, s')$
if $s' \succ_L s$ and 
$s \succ_c s' \succ_c \emptyset$ holds. 
Otherwise, we say $c$ agrees with $L$ for $(s, s')$. 
In short, 
$c$ disagrees with $L$ for $(s, s')$,
when $s'$ is ranked higher than $s$ in $L$, 
both $s$ and $s'$ are acceptable for college $c$, 
and $c$ prefers $s$ over $s'$. 
Assume we use SD based on $L$. 
Then, in obtained matching $Y$, 
if $c$ disagrees with $L$ for $(s, s')$, 
$s$ has a chance to have justified envy toward $s'$ in $c$, 
since $s'$ is chosen before $s$ and can be allocated to $c$, 
while $s$ might not be allocated to $c$. 
On the other hand, if 
$c$ agrees with $L$ for $(s, s')$,
then $s$ never has justified envy toward $s'$ in $c$. 
This is because, the fact that 
$c$ agrees with $L$ for $(s, s')$ means: 
(i) $s$ is ranked higher than $s'$ in $L$, 
(ii) $s$ is ranked lower than $s'$ in $c$, 
or (iii) either $s$ or $s'$ is unacceptable for $c$. 
In each of the above three cases, $s$ cannot have 
justified envy toward $s'$ in $c$. 

Let $d(L,s)$ denote
$|\{s' \mid s'\in S\setminus\{s\}, c \in C, 
c \text{ disagrees with } L \text{ for }\linebreak (s, s')\}|$, 
i.e., $d(L,s)$ counts the number of students such that 
for some college $c$, a disagreement related to $s$ occurs. 

The following theorem holds. 
\begin{theorem}
 \label{thm:ML}   Assume for master-list $L$, 
$\forall s \in S$, $d(L,s)\leq k$ holds. 
Then, SD using $L$ is EF-$k$.
\end{theorem}
\begin{proof}
Assume, for the sake of contradiction, 
that in obtained matching $Y$, there exists student $s$ with $|Ev(Y, s)| > k$.
Then, for each $s' \in Ev(Y,s)$, 
we have 
(i) $s' \succ_L s$, and 
(ii) for $(s', c) \in Y$, $s \succ_c s' \succ_c \emptyset$. 
Thus, $c$ disagrees $L$ for $(s, s')$. 
This is true for each $s' \in Ev(Y,s)$.
Then, $d(L,s) > k$ holds, a contradiction. 
\end{proof}

Theorem~\ref{thm:ML} means that if we can choose a good master-list $L$, 
such that $\max_{s \in S} d(L,s)$ is small, e.g. at most $k$, the obtained matching is 
guaranteed to be EF-$k$. 
Note that this guarantee holds independently from the actual distributional constraints and 
students' preferences; $k$ can be computed using colleges' preference profile $\succ_C$ only. 
Thus, for given students' preference $\succ_S$,
the obtained matching can be EF-$k'$ for $k'$ that is much smaller than $k$ guaranteed by Theorem 6.1; see the experimental results that clarify this in the next section.

Let us examine the problem of finding an optimal master-list 
(in terms of minimizing $\max_{s \in S} d(L,s)$) 
for given colleges' preference profile $\succ_C$. 
\begin{theorem}
\label{thm:choosing-ML}
For given colleges' preference profile $\succ_C$, 
computing master-list $L$, which minimizes $\max_{s \in S} d(L,s)$
can be done in polynomial time. 
\end{theorem}
\ifconference
\else
\begin{proof}
Let us first introduce a graphical representation of the above optimization problem. 
Consider a directed graph $G=(S,E)$, where each student is a vertex. 
For a pair of students $s$ and $s'$, if there exists college $c$ s.t. 
$s \succ_c s' \succ_c \emptyset$ holds, we add a directed edge 
$(s, s')$. This means that to make $c$ agree with the obtained master-list for 
$(s, s')$, the master-list must rank $s$ higher than $s'$. 
Then, for $G=(S,E)$ and master-list $L$, 
$d(L,s)$ is equal to the number of outgoing edges from $s$ toward 
any of higher-ranked students than $s$ in $L$. 
For $s$, let $O_s$ denote the set of outgoing edges from $s$. 
Clearly, for any $L$, $d(L,s) \leq |O_s|$ holds. 
Also, $d(L,s) = |O_s|$ holds when $s$ is ranked lowest in $L$. 
This implies $\max_{s \in S} d(L,s) \geq \min_{s\in S} |O_s|$ holds, 
i.e., the optimal $k$ cannot be smaller than 
$\min_{s\in S} |O_s|$. 
This is because some student $s$ must be ranked lowest in $L$, and 
$d(L,s)=|O_s|$ holds. 
Then, when choosing the student who should be ranked lowest in $L$, 
we can safely choose $s$ with the smallest $|O_s|$ to guarantee $L$'s optimality. 
Thus, the following greedy algorithm obtains an optimal master-list $L$ 
(as well as $\max_{s \in S} d(L, s)$).
\begin{enumerate}
    \item For given graph $G=(V,E)$ (where $V=S$), 
    set $k \leftarrow 0$, and $L$ to an empty list. 
    \item If $V=\emptyset$, return $L$ and $k$. 
    \item Choose $s = \arg \min_{s \in V} |O_s|$. Add $s$ to the top of $L$.
    $k\leftarrow \max(k, |O_s|)$.
    Remove $s$ and all edges related to $s$ from $G$. Go to (2). 
\end{enumerate}
Clearly, the complexity of this greedy algorithm is 
$O(|V||E|)$. 
\end{proof}
\fi

Let us call SD mechanism using optimal $L$ as SD$^*$. 
SD$^*$ is strategyproof and Pareto efficient. 
When we apply SD$^*$ to the matching instance 
presented in Example~\ref{ex:regional}, 
the above algorithm returns $L$ with
$\max_{s \in S} d(L,s)=n-1$ and 
the obtained matching cannot be
EF-$k$ for any $k< n-1$. 
In the next section, we show that 
SD$^*$ can be EF-$k$ for smaller $k$ when 
colleges' preferences are similar. 

Let us examine situations where 
SD$^*$ can be used in practice. 
Assume there exists an authority who decides a matching based on colleges'/students' preferences. The authority is allowed to override colleges' preferences to some extent in order to improve students' welfare. More specifically, the authority can use its own ordering among students to decide the matching, where the ordering is chosen such that it is as close as possible to colleges' preferences. Our SD$^*$ is based on this idea, which uses ordering $L$ that minimizes $k=\max_{s \in S} d(L,s)$.
The obtained matching is guaranteed to be EF-$k$. There can be alternative minimization criteria for choosing $L$, e.g., minimizing the sum of Kendall tau distances (the number of pairwise disagreements). However, this optimization problem is computationally hard~\cite{bartholdi:1989} and can be \emph{unfair} among students.

\subsection{EF-$k$ mechanism}
Next, we develop a strategyproof and EF-$k$ mechanism 
for any given $k\leq n-1$.
First, let us define the standard Sample and Deferred 
Acceptance (SDA) mechanism. This mechanism is 
developed by \citet{liu:spr2023} for a special case for hereditary constraints where the maximum quota 
of each college is determined by allocating indivisible 
resources to each college. 
The basic idea of SDA is to combine SD and ACDA. 
One major limitation of ACDA is that we need to determine the maximal feasible vector 
$\nu^*$ (which determines the maximum  quota of each college)
independently from students' preferences. As a result, the maximum quotas of popular colleges can be low, 
while those of unpopular colleges can be high. 
In the standard SDA, 
first, we choose a subset of students $S' \subseteq S$, 
where $|S'|=k$. We call $S'$ \emph{sampled} students, and 
$S\setminus S'$ \emph{regular} students. 
We assign sampled students using SD. 
Assume the obtained matching for sampled students is $Y'$. 
Then, we choose a  maximal feasible vector $\nu^*$ based on 
the preferences of sampled students. 
\citet{liu:spr2023} present several 
alternative ways to choose $\nu^*$. In this paper, as described later, 
we apply a simulation-based method using copies of sampled students, which is shown to be most effective in \cite{liu:spr2023}.
Then, we apply ACDA for regular students, where 
maximum quota $q_{c_i}$ for each college $c_i$ is given as 
$\nu^*_i - |Y'_{c_i}|$. 

The standard SDA is strategyproof. It is also EF-$k$, since 
for each sampled student $s$, she has justified envy 
only toward another sampled student assigned before her. 
Thus, $|Ev(Y,s)|\leq k-1$ holds. 
Also, since DA is fair, for regular student $s$, 
she has justified envy only toward sampled students. 
Thus, $|Ev(Y,s)|\leq k$ holds. 

However, if the preferences of sampled students are 
completely different from the preferences of regular 
students, obtained $\nu^*$ can be bad for regular 
students. As a result, even no vacant college property
is not satisfied. 
We can assume SDA satisfies no empty matching property.
No empty matching property is violated only in an exceptional 
case where all sampled students assume all colleges unacceptable. 
In such a case, we can choose additional sampled students until 
at least one student is assigned to some college. 

We propose a generalized version of SDA, such that 
no vacant college property is satisfied under a mild 
assumption. The basic idea is that, since 
there exists a chance that 
the preferences of sampled students are 
completely different from those of  regular 
students, we reserve some seats for each college 
even if the college seems unpopular based on 
the preferences of sampled students. 
Let $\widehat{\nu} = (\widehat{\nu}_1, 
\ldots, \widehat{\nu}_m)$ be 
\emph{reserved quotas}, where 
$\widehat{\nu}_i \geq 0$ for each $i \in M$, 
and $f(\widehat{\nu})=0$ holds. 
The goal of the reserved quotas $\widehat{\nu}$
is to guarantee that each college $c_i$ is guaranteed 
to accept at least $\widehat{\nu}_i$ students, as long 
as enough students hope to be assigned to $c_i$, 
even if $c_i$ seems unpopular among sampled students. 

For two $m$-element vectors $\nu$ and $\nu'$, 
let $\nu \vee \nu'$ denote the element-wise maximum, 
i.e., $\nu \vee \nu' = (\max(\nu_1, \nu'_1), 
\ldots, \max(\nu_m, \nu'_m))$. 

First, let us define SD with reserved quotas $\widehat{\nu}$. 
As standard SD, we assign students sequentially based 
on master-list $L$. Let $Y$ denote the assignment obtained 
so far. The current student can be assigned to 
$c_i$, as long as $f((\nu(Y) + \cvec{i})\vee \widehat{\nu})=0$ holds. In short, the current student
$s$ can be assigned to $c_i$, if  
$c_i$ can still accept one more student, assuming each college 
$c_j$ will be assigned at least $\widehat{\nu}_j$ students.

Then, SDA with reserved quotas $\widehat{\nu}$ is 
defined as follows. 
Choose $k$ sampled students (the remaining students are regular students). 
They are assigned by SD with reserved quotas $\widehat{\nu}$. 
Let $Y'$ denote the matching for sampled students. 
Then, obtain a matching $Y''$, by further assigning 
multiple virtual students, each of which is a copy of sampled students by SD with
reserved quotas, until no more student can be assigned. 
More specifically, let us assume sampled students are $s_1, \ldots, s_k$. 
We create virtual students ${s}_{i,1},{s}_{i,2},\ldots$, 
which are copies of each sampled student ${s}_i$.
Then, after sampled students are assigned. We assign these virtual students in a round-robin order, i.e., 
${s}_{1,1}, {s}_{2,1}, \ldots, {s}_{k,1}, 
 {s}_{1,2}, {s}_{2,2}, \ldots, {s}_{k,2}, 
  {s}_{1,3}, {s}_{2,3}, \ldots, {s}_{k,3}, \ldots$. 
Note that this procedure is just for choosing appropriate $\nu^*$; in reality, 
these virtual students are not allocated to 
any college. 
Then, we choose maximal feasible vector 
$\nu^*$ such that $\nu^* \geq \nu(Y'') \vee \widehat{\nu}$ holds.
For each college $c_i$, we set its maximum quota
$q_{c_i}$ as $\nu^*_i - |Y'_{c_i}|$, and 
run ACDA for regular students. 

\begin{theorem}
Assume for $\widehat{\nu}$, 
$f(\widehat{\nu})=0$ holds, and 
$\forall i \in M$, such that $f(\cvec{i})=0$ holds, 
$\widehat{\nu}_i \geq 1$ also holds.
Then, SDA with reserved quotas $\widehat{\nu}$ and $k$-sampled students is 
strategyproof, EF-$k$, and satisfies no vacant college property. 
\end{theorem}
\begin{proof}
It is clear even after the above modifications, 
SDA with reserved quotas $\widehat{\nu}$ is still 
strategyproof and EF-$k$. 

We show that it also satisfies no vacant college property. 
Assume, for the sake of contradiction, that obtained matching 
$Y$ does not satisfy no vacant college property, 
i.e., student $s$ strongly claims an empty seat of $c_i$, 
while $Y_s=\emptyset$ and $Y_{c_i}=\emptyset$. 
Since $Y$ is obtained by SDA with reserved quotas 
$\nu^*$, 
$\nu(Y) \leq \nu^*$ holds. 
Also, $Y_{c_i}=\emptyset$ and $\nu^*_i \geq \widehat{\nu}_i 
\geq 1$ holds. 
However, this fact means that if 
$s$ applies to $c_i$, she must be accepted to $c_i$ 
(either in SD with reserved quotas or ACDA). This violates the fact that $Y_s=\emptyset$.
\end{proof}

Let us examine situations where 
SDA can be used in practice. 
Assume there exist $k$ 
\emph{distinguished} students, e.g., they have excellent achievements in sports / volunteer works, etc., they are from financially difficult families / minority groups, or even chosen by lottery. If giving them priority in college administration is socially acceptable, we can use these \emph{distinguished} students as sampled students in SDA. 
Then, the outcome is guaranteed to be EF-$k$. 

\section{Experimental Evaluation}
\label{sec:evaluation}
\def\ratio{0.229}
\def\spaceratio{0.1}
\def\trimspacebetweencaption{\vspace{-8mm}}
\def\trimspaceaftercaption{\vspace{-4mm}}
\begin{figure*}[t]
	    \begin{minipage}[t]{\ratio\textwidth}
	        \begin{center}
         \vspace{\ratio\textwidth}
        \includegraphics[scale=0.14]{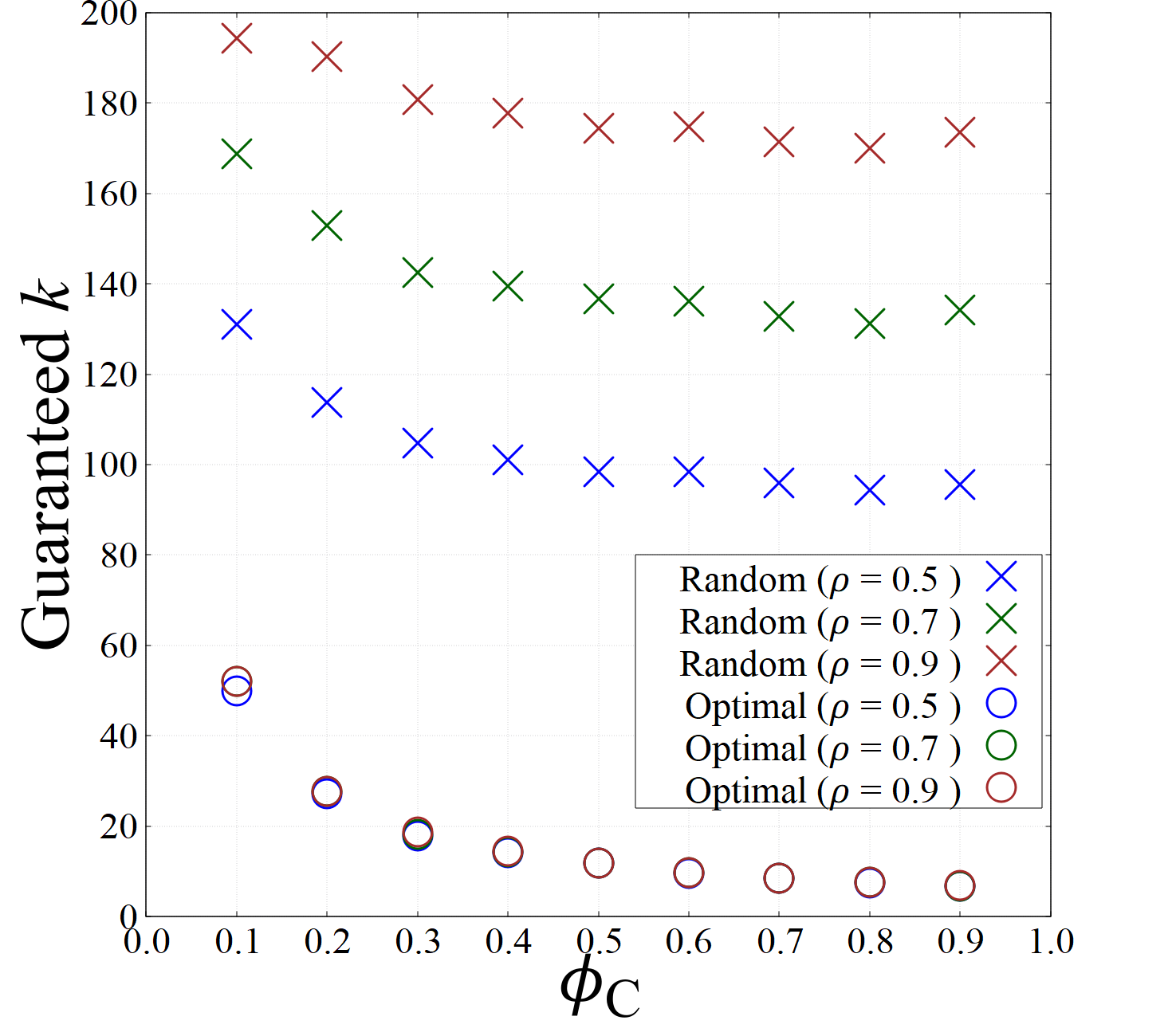}
\trimspacebetweencaption
	        \caption{Guaranteed $k$ for optimal/random master-list}
	        \label{fig:sd-k-guaranteed}
	        \end{center}
     \end{minipage}
     \hspace{\spaceratio\textwidth}
	    \begin{minipage}[t]{\ratio\textwidth}
	        \begin{center}
      \vspace{\ratio\textwidth}
        \includegraphics[scale=0.25]{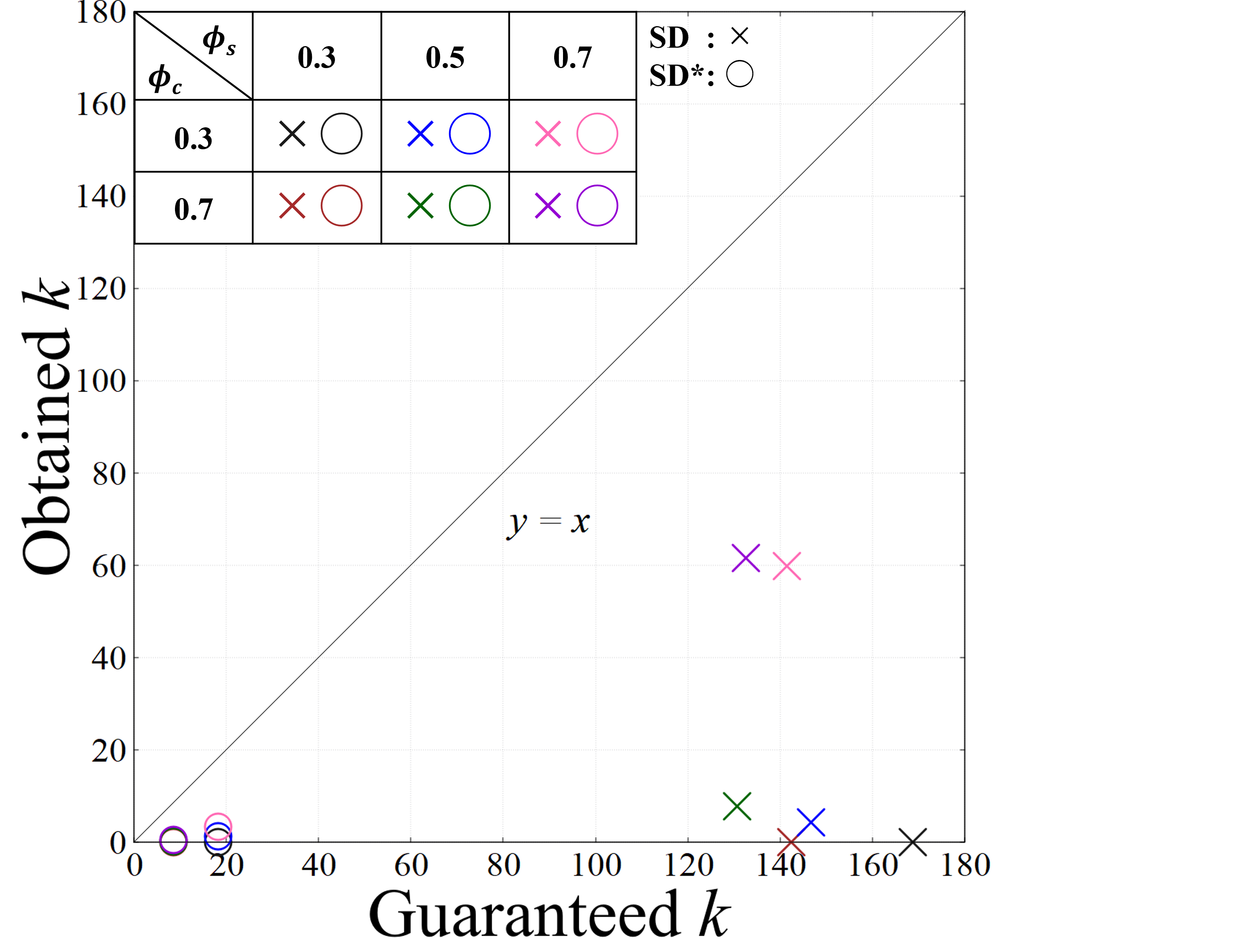}
        \trimspacebetweencaption
	        \caption{Comparison between obtained/guaranteed $k$ for SD$^*$/SD}
	        \label{fig:sd-k-actual}
	        \end{center}         
	    \end{minipage}
     \hspace{\spaceratio\textwidth}     
	    \begin{minipage}[t]{\ratio\textwidth}
	        \begin{center}
      \vspace{\ratio\textwidth}
      \includegraphics[scale=0.14]{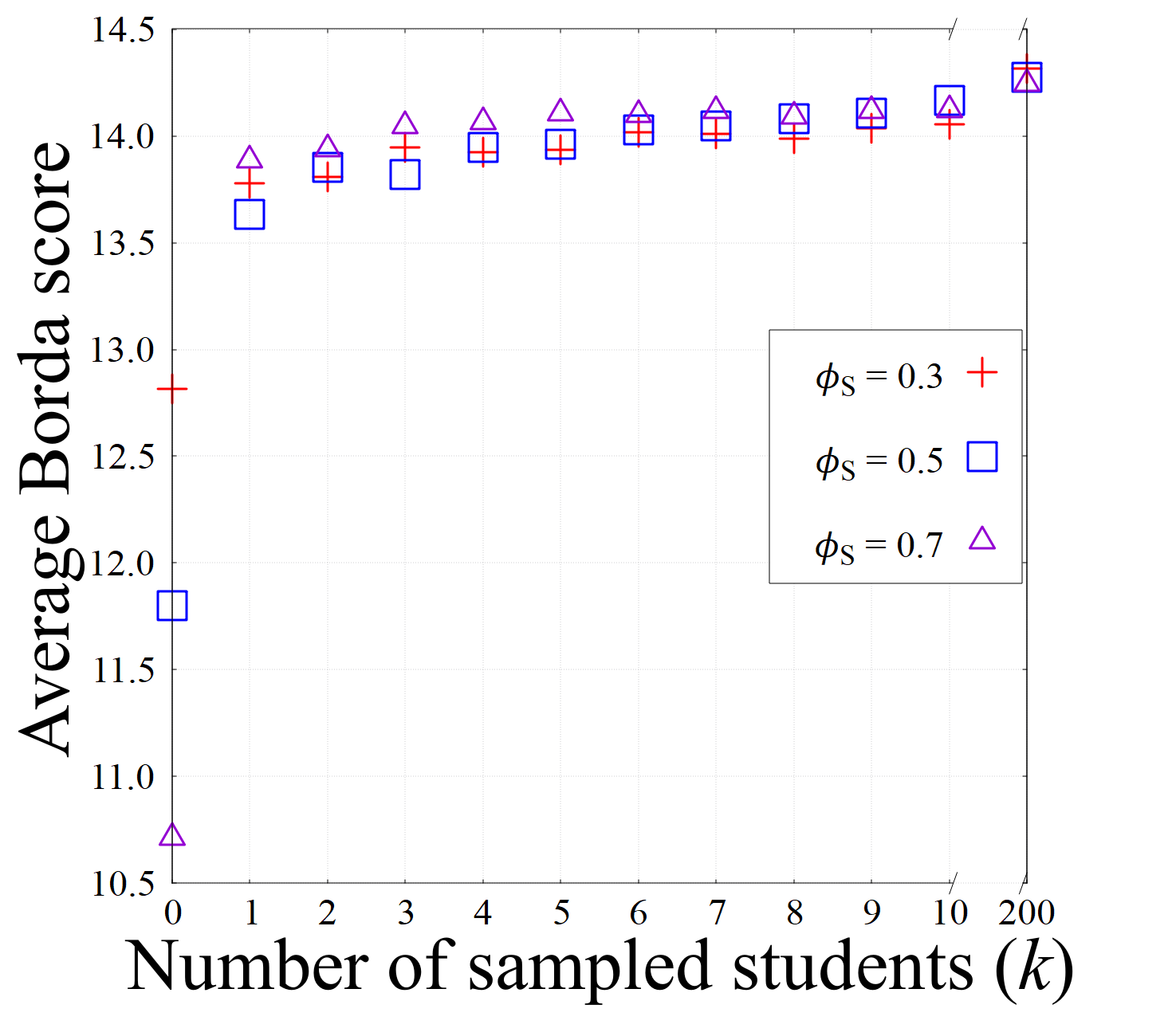}
      \trimspacebetweencaption
	        \caption{Average Borda score for SDA with reserved quotas}
	        \label{fig:sda-k-borda}
	        \end{center}         
	    \end{minipage}
\trimspaceaftercaption
 \end{figure*}
First, we show the level of $k$ that SD$^*$ can be guaranteed by using an optimal master-list.
We set the number of students $n$ to 200 and 
the number of colleges $m$ to 20. 
We generate the preference of each college $c$ using the Mallows model%
~\cite{drummond:ijcai:2013,lu:jmlr:2014,mallows1957non};
college preference $\succ_c$ is drawn with probability:
$\Pr(\succ_c)=\frac{\exp (-\phi_C \cdot \delta(\succ_c,
	\succ_{\widehat{c}}))}{\sum_{\succ'_c} \exp(-\phi_c \cdot \delta(\succ'_c,
	\succ_{\widehat{c}}))}.$
Here $\phi_C \in \mathbf{R}_{+}$ denotes the spread parameter for colleges, $\succ_{\widehat{c}}$ is a
central preference uniformly randomly chosen from all possible
preferences,
and $\delta(\succ_c, \succ_{\widehat{c}})$ represents the Kendall tau
distance, which is the number of pairwise inversions between $\succ_c$ and $\succ_{\widehat{c}}$.
Intuitively, colleges' preferences are distributed around a central
preference with
spread parameter $\phi_C$. 
When $\phi_C= 0$, the Mallows model becomes identical to the uniform distribution,
while increasing $\phi_C$ leads to convergence towards a constant
distribution, yielding $\succ_{\widehat c}$. 
Initially, each $\succ_c$ does not include
$\emptyset$. We insert $\emptyset$ at the position
$\lfloor \rho \cdot n \rfloor$ (where $0 < \rho < 1$).

Figure~\ref{fig:sd-k-guaranteed} shows the guaranteed $k$ when using 
an optimal master-list by varying the
spread parameter $\phi_C$ and $\rho$. 
Each data point is an average of 10 instances. 
We also show the result when the master-list
is randomly chosen.
We can see that when the spread parameter becomes larger 
(colleges' preferences become more similar), 
SD$^*$ can guarantee EF-$k$ for smaller $k$.
For example, $k$ becomes less than 5\% of $n$ when $\phi_C$ is
0.6. We can see $\rho$ has almost no effect on SD$^*$, while
it significantly affects randomly selected master-lists. 

Next, we apply SD$^*$ to each matching market and measure
the obtained level of $k$ that SD$^*$ achieves.
We consider the following distributional constraints~\cite{liu:spr2023}.
There exists a set of indivisible resources
$R=\{r_1, \ldots, r_{|R|}\}$. Each resource $r$ has its capacity
$q_r \in\mathbb{N}_{>0}$.
For each resource $r$, its college compatibility list $T_r$ is defined;
resource $r$ can be allocated to exactly one college
in $T_r\subseteq C$. 
Mapping $\mu$ denotes one possible allocation of resources to colleges,
i.e., $\mu:R\rightarrow C$ maps each resource $r$ to a college
$\mu(r) \in T_r$.
For given allocation $\mu$, the maximum quota of college $c$ is given
as $q_{\mu}(c) = \sum_{r: \mu(r)=c} q_r$, i.e., the maximum quota of each college is endogenously determined as the sum of the capacities of
allocated resources. 
We assume $f(\nu)=0$ if there exists $\mu$ s.t.
$\nu_i \leq q_{\mu}(c_i)$ holds for all $i \in M$.
Each market has $|R| = 100$ resources.
For each resource $r$, we generate $T_r$ such that
each college $c$ is included in $T_r$ with probability $0.3$.
There are 40, 20, and 40 resources with capacity 1, 2, and 3, respectively; thus 
the total capacity of colleges is equal to $n$. 
We generate each student's preference in a similar way as a college's preference, i.e., 
we utilize the Mallows model with spread parameter $\phi_S$.
We do not apply $\rho$ for students; each student considers all colleges acceptable.

Figure~\ref{fig:sd-k-actual} shows the average of 10 instances.
The $x$-axis shows the guaranteed $k$ and the $y$-axis shows the actually obtained $k$. We set colleges' spread parameter $\phi_C$ to $0.3$ and $0.7$, 
and students' spread parameter $\phi_S$ to $0.3$, $0.5$, and $0.7$.
$\rho$ is set to 0.7. 
By definition, each data point must be located in the lower-right half.
The result shows the actually obtained $k$ is much smaller than
the guaranteed $k$.
In particular, for SD$^*$, it is between 0 and 4. 
For SD, we can see that when $\phi_S$ becomes larger,
the competition among students becomes more intense.
As a result, more students tend to have justified envy. 

Next, we evaluate SDA with reserved quotas. By varying $k$,
it can be identical to ACDA (when $k=0$) and SD$^*$ (when
$k=n$), assuming we use the same master-list as SD$^*$ and the same reserved quotas.
Figure~\ref{fig:sda-k-borda} shows 
the average Borda score of the students
varying $k$ and the students' spread parameter $\phi_S$.
If a student is assigned to her $i$-th choice
college, her Borda score is $m-i+1$.
We fix the colleges' spread parameter $\phi_C$ to 0.7 and $\rho$ to 0.7.
We set reserved quotas $\widehat{\nu}$ to $(1, 1, \ldots, 1)$.
Each data point represents an average of 10 instances. 
In this setting, SDA with $n$ sampled students (which is identical to SD$^*$) 
guarantees EF-$k$ for $k=9$ in average. 
%
%
The average Borda score significantly improves 
as $k$ increases from the case where $k = 0$.
Note that increasing the average Borda score by one is significant; each student must be assigned to a strictly better college. 
The difference between $k=0$ (where SDA is identical to ACDA) and $k=1$
becomes larger when $\phi_S$ becomes larger, i.e., 
when students' preferences are similar. 
We can see that SDA achieves a high degree of fairness and 
efficiency with a few sampled students.

In summary, SD$^*$ is much fairer than SD with a randomly selected master-list, and can attain EF-$k'$ for $k'$ that is much smaller than $k$ guaranteed by Theorem 6.1.
Also, SDA with reserved quotas is much more efficient than ACDA, and attains very good fairness at the expense of a little efficiency compared to SD*.
\section{Conclusions and future works}
When distributional constraints are imposed in two-sided matching, there exists a trade-off 
between fairness and efficiency. 
We clarified the tight boundaries 
on whether a strategyproof and fair mechanism can satisfy certain
efficiency properties for each class of constraints. 
We also established a new fairness requirement called EF-$k$. 
We examined theoretical properties related to EF-$k$, and 
developed two contrasting strategyproof mechanisms 
that work for general hereditary constraints. 
We evaluated the performance of these mechanisms via computer simulation.
We believe EF-$k$ is significant since it brings up many new research topics in constrained matching literature; there remain many open questions related to EF-$k$. 
For example, can any strategyproof mechanism guarantee 
EF-$k$ for some fixed $k$ in conjunction with some efficiency property (which is stronger than 
no vacant college property, e.g., weak nonwastefulness)? 
Furthermore, 
there exists another mechanism called Adaptive DA \cite{goto:17} that works for any hereditary constraints. 
Comparing this mechanism with our newly proposed mechanisms is our immediate future work.

\begin{acks}
We would like to thank anonymous reviewers for their valuable comments. 
This work was partially supported by JST ERATO Grant Number JPMJER2301, and JSPS
KAKENHI Grant Numbers JP21H04979 and JP20H00609, Japan.
\end{acks}

\clearpage
\bibliographystyle{ACM-Reference-Format} 
\bibliography{matching}


\begin{thebibliography}{32}


\ifx \showCODEN    \undefined \def \showCODEN     #1{\unskip}     \fi
\ifx \showDOI      \undefined \def \showDOI       #1{#1}\fi
\ifx \showISBNx    \undefined \def \showISBNx     #1{\unskip}     \fi
\ifx \showISBNxiii \undefined \def \showISBNxiii  #1{\unskip}     \fi
\ifx \showISSN     \undefined \def \showISSN      #1{\unskip}     \fi
\ifx \showLCCN     \undefined \def \showLCCN      #1{\unskip}     \fi
\ifx \shownote     \undefined \def \shownote      #1{#1}          \fi
\ifx \showarticletitle \undefined \def \showarticletitle #1{#1}   \fi
\ifx \showURL      \undefined \def \showURL       {\relax}        \fi
\providecommand\bibfield[2]{#2}
\providecommand\bibinfo[2]{#2}
\providecommand\natexlab[1]{#1}
\providecommand\showeprint[2][]{arXiv:#2}

\bibitem[\protect\citeauthoryear{Aziz, Baychkov, and Bir{\'o}}{Aziz
  et~al\mbox{.}}{2022a}]%
        {aziz:cutoff:2021}
\bibfield{author}{\bibinfo{person}{Haris Aziz}, \bibinfo{person}{Anton
  Baychkov}, {and} \bibinfo{person}{Peter Bir{\'o}}.}
  \bibinfo{year}{2022}\natexlab{a}.
\newblock \showarticletitle{Cutoff stability under distributional constraints
  with an application to summer internship matching}.
\newblock \bibinfo{journal}{\emph{Mathematical Programming}}
  (\bibinfo{year}{2022}).
\newblock
\urldef\tempurl%
\url{https://doi.org/10.1007/s10107-022-01917-1}
\showDOI{\tempurl}


\bibitem[\protect\citeauthoryear{Aziz, Bir{\'{o}}, Fleiner, Gaspers, de~Haan,
  Mattei, and Rastegari}{Aziz et~al\mbox{.}}{2022c}]%
        {aziz2022stable}
\bibfield{author}{\bibinfo{person}{Haris Aziz}, \bibinfo{person}{P{\'{e}}ter
  Bir{\'{o}}}, \bibinfo{person}{Tam{\'{a}}s Fleiner}, \bibinfo{person}{Serge
  Gaspers}, \bibinfo{person}{Ronald de Haan}, \bibinfo{person}{Nicholas
  Mattei}, {and} \bibinfo{person}{Baharak Rastegari}.}
  \bibinfo{year}{2022}\natexlab{c}.
\newblock \showarticletitle{Stable matching with uncertain pairwise
  preferences}.
\newblock \bibinfo{journal}{\emph{Theoretical Computer Science}}
  \bibinfo{volume}{909} (\bibinfo{year}{2022}), \bibinfo{pages}{1--11}.
\newblock
\urldef\tempurl%
\url{https://doi.org/10.1016/j.tcs.2022.01.028}
\showDOI{\tempurl}


\bibitem[\protect\citeauthoryear{Aziz, Bir\'{o}, and Yokoo}{Aziz
  et~al\mbox{.}}{2022b}]%
        {ABY:aaai22:survery}
\bibfield{author}{\bibinfo{person}{Haris Aziz}, \bibinfo{person}{Peter
  Bir\'{o}}, {and} \bibinfo{person}{Makoto Yokoo}.}
  \bibinfo{year}{2022}\natexlab{b}.
\newblock \showarticletitle{Matching Market Design with Constraints}. In
  \bibinfo{booktitle}{\emph{Proceedings of the 36th AAAI Conference on
  Artificial Intelligence (AAAI-2022)}}. \bibinfo{pages}{12308--12316}.
\newblock
\urldef\tempurl%
\url{https://doi.org/10.1609/aaai.v36i11.21495}
\showDOI{\tempurl}


\bibitem[\protect\citeauthoryear{Aziz, Gaspers, Sun, and Walsh}{Aziz
  et~al\mbox{.}}{2019}]%
        {Haris19matching}
\bibfield{author}{\bibinfo{person}{Haris Aziz}, \bibinfo{person}{Serge
  Gaspers}, \bibinfo{person}{Zhaohong Sun}, {and} \bibinfo{person}{Toby
  Walsh}.} \bibinfo{year}{2019}\natexlab{}.
\newblock \showarticletitle{From Matching with Diversity Constraints to
  Matching with Regional Quotas}. In \bibinfo{booktitle}{\emph{Proceedings of
  the 18th International Conference on Autonomous Agents and MultiAgent Systems
  (AAMAS-2019)}}. \bibinfo{pages}{377--385}.
\newblock


\bibitem[\protect\citeauthoryear{Bartholdi, Tovey, and Trick}{Bartholdi
  et~al\mbox{.}}{1989}]%
        {bartholdi:1989}
\bibfield{author}{\bibinfo{person}{J. Bartholdi}, \bibinfo{person}{C.~A.
  Tovey}, {and} \bibinfo{person}{M.A. Trick}.} \bibinfo{year}{1989}\natexlab{}.
\newblock \showarticletitle{Voting Schemes for which It Can Be Difficult to
  Tell Who Won the Election}.
\newblock \bibinfo{journal}{\emph{Social Choice and Welfare}}
  \bibinfo{volume}{6} (\bibinfo{year}{1989}), \bibinfo{pages}{157--165}.
\newblock
\urldef\tempurl%
\url{https://doi.org/10.1007/BF00303169}
\showDOI{\tempurl}


\bibitem[\protect\citeauthoryear{Budish}{Budish}{2011}]%
        {budish2011combinatorial}
\bibfield{author}{\bibinfo{person}{Eric Budish}.}
  \bibinfo{year}{2011}\natexlab{}.
\newblock \showarticletitle{The combinatorial assignment problem: Approximate
  competitive equilibrium from equal incomes}.
\newblock \bibinfo{journal}{\emph{Journal of Political Economy}}
  \bibinfo{volume}{119}, \bibinfo{number}{6} (\bibinfo{year}{2011}),
  \bibinfo{pages}{1061--1103}.
\newblock
\urldef\tempurl%
\url{https://doi.org/10.1086/664613}
\showDOI{\tempurl}


\bibitem[\protect\citeauthoryear{Cho, Koshimura, Mandal, Yahiro, and Yokoo}{Cho
  et~al\mbox{.}}{2022}]%
        {cho:2022}
\bibfield{author}{\bibinfo{person}{Sung-Ho Cho}, \bibinfo{person}{Miyuki
  Koshimura}, \bibinfo{person}{Pinaki Mandal}, \bibinfo{person}{Kentaro
  Yahiro}, {and} \bibinfo{person}{Makoto Yokoo}.}
  \bibinfo{year}{2022}\natexlab{}.
\newblock \showarticletitle{Impossibility of weakly stable and strategy-proof
  mechanism}.
\newblock \bibinfo{journal}{\emph{Economics Letters}}  \bibinfo{volume}{217}
  (\bibinfo{year}{2022}), \bibinfo{pages}{110675}.
\newblock
\urldef\tempurl%
\url{https://doi.org/10.1016/j.econlet.2022.110675}
\showDOI{\tempurl}


\bibitem[\protect\citeauthoryear{Drummond and Boutilier}{Drummond and
  Boutilier}{2013}]%
        {drummond:ijcai:2013}
\bibfield{author}{\bibinfo{person}{Joanna Drummond} {and}
  \bibinfo{person}{Craig Boutilier}.} \bibinfo{year}{2013}\natexlab{}.
\newblock \showarticletitle{Elicitation and Approximately Stable Matching with
  Partial Preferences}. In \bibinfo{booktitle}{\emph{Proceedings of the 23rd
  International Joint Conference on Artificial Intelligence (IJCAI-2013)}}.
  \bibinfo{pages}{97--105}.
\newblock


\bibitem[\protect\citeauthoryear{Ehlers, Hafalir, Yenmez, and Yildirim}{Ehlers
  et~al\mbox{.}}{2014}]%
        {ehlers::2012}
\bibfield{author}{\bibinfo{person}{Lars Ehlers}, \bibinfo{person}{Isa~E.
  Hafalir}, \bibinfo{person}{M.~Bumin Yenmez}, {and}
  \bibinfo{person}{Muhammed~A. Yildirim}.} \bibinfo{year}{2014}\natexlab{}.
\newblock \showarticletitle{School Choice with Controlled Choice Constraints:
  Hard Bounds versus Soft Bounds}.
\newblock \bibinfo{journal}{\emph{Journal of Economic Theory}}
  \bibinfo{volume}{153} (\bibinfo{year}{2014}), \bibinfo{pages}{648--683}.
\newblock
\urldef\tempurl%
\url{https://doi.org/10.1016/j.jet.2014.03.004}
\showDOI{\tempurl}


\bibitem[\protect\citeauthoryear{Fragiadakis, Iwasaki, Troyan, Ueda, and
  Yokoo}{Fragiadakis et~al\mbox{.}}{2015}]%
        {fragiadakis::2012}
\bibfield{author}{\bibinfo{person}{Daniel Fragiadakis},
  \bibinfo{person}{Atsushi Iwasaki}, \bibinfo{person}{Peter Troyan},
  \bibinfo{person}{Suguru Ueda}, {and} \bibinfo{person}{Makoto Yokoo}.}
  \bibinfo{year}{2015}\natexlab{}.
\newblock \showarticletitle{Strategyproof Matching with Minimum Quotas}.
\newblock \bibinfo{journal}{\emph{ACM Transactions on Economics and
  Computation}} \bibinfo{volume}{4}, \bibinfo{number}{1}
  (\bibinfo{year}{2015}), \bibinfo{pages}{6:1--6:40}.
\newblock
\urldef\tempurl%
\url{https://doi.org/10.1145/2841226}
\showDOI{\tempurl}


\bibitem[\protect\citeauthoryear{Gale and Shapley}{Gale and Shapley}{1962}]%
        {Gale:AMM:1962}
\bibfield{author}{\bibinfo{person}{David Gale} {and}
  \bibinfo{person}{Lloyd~Stowell Shapley}.} \bibinfo{year}{1962}\natexlab{}.
\newblock \showarticletitle{College Admissions and the Stability of Marriage}.
\newblock \bibinfo{journal}{\emph{The American Mathematical Monthly}}
  \bibinfo{volume}{69}, \bibinfo{number}{1} (\bibinfo{year}{1962}),
  \bibinfo{pages}{9--15}.
\newblock
\urldef\tempurl%
\url{https://doi.org/10.2307/2312726}
\showDOI{\tempurl}


\bibitem[\protect\citeauthoryear{Gibbard}{Gibbard}{1973}]%
        {gibbard:1973}
\bibfield{author}{\bibinfo{person}{Allan Gibbard}.}
  \bibinfo{year}{1973}\natexlab{}.
\newblock \showarticletitle{Manipulation of Voting Schemes: A General Result}.
\newblock \bibinfo{journal}{\emph{Econometrica}} \bibinfo{volume}{41},
  \bibinfo{number}{4} (\bibinfo{year}{1973}), \bibinfo{pages}{587--601}.
\newblock
\urldef\tempurl%
\url{https://doi.org/10.2307/1914083}
\showDOI{\tempurl}


\bibitem[\protect\citeauthoryear{Goto, Kojima, Kurata, Tamura, and Yokoo.}{Goto
  et~al\mbox{.}}{2017}]%
        {goto:17}
\bibfield{author}{\bibinfo{person}{Masahiro Goto}, \bibinfo{person}{Fuhito
  Kojima}, \bibinfo{person}{Ryoji Kurata}, \bibinfo{person}{Akihisa Tamura},
  {and} \bibinfo{person}{Makoto Yokoo.}} \bibinfo{year}{2017}\natexlab{}.
\newblock \showarticletitle{Designing Matching Mechanisms under General
  Distributional Constraints}.
\newblock \bibinfo{journal}{\emph{American Economic Journal: Microeconomics}}
  \bibinfo{volume}{9}, \bibinfo{number}{2} (\bibinfo{year}{2017}),
  \bibinfo{pages}{226--262}.
\newblock
\urldef\tempurl%
\url{https://doi.org/10.1257/mic.20160124}
\showDOI{\tempurl}


\bibitem[\protect\citeauthoryear{Hatfield and Milgrom}{Hatfield and
  Milgrom}{2005}]%
        {Hatfield:AER:2005}
\bibfield{author}{\bibinfo{person}{John~William Hatfield} {and}
  \bibinfo{person}{Paul~R. Milgrom}.} \bibinfo{year}{2005}\natexlab{}.
\newblock \showarticletitle{Matching with Contracts}.
\newblock \bibinfo{journal}{\emph{American Economic Review}}
  \bibinfo{volume}{95}, \bibinfo{number}{4} (\bibinfo{year}{2005}),
  \bibinfo{pages}{913--935}.
\newblock
\urldef\tempurl%
\url{https://doi.org/10.1257/0002828054825466}
\showDOI{\tempurl}


\bibitem[\protect\citeauthoryear{Hosseini, Larson, and Cohen}{Hosseini
  et~al\mbox{.}}{2015}]%
        {hosseini2015manipulablity}
\bibfield{author}{\bibinfo{person}{Hadi Hosseini}, \bibinfo{person}{Kate
  Larson}, {and} \bibinfo{person}{Robin Cohen}.}
  \bibinfo{year}{2015}\natexlab{}.
\newblock \showarticletitle{On Manipulablity of Random Serial Dictatorship in
  Sequential Matching with Dynamic Preferences}. In
  \bibinfo{booktitle}{\emph{Proceedings of the 29th {AAAI} Conference on
  Artificial Intelligence (AAAI-2015)}}. \bibinfo{pages}{4168--4169}.
\newblock
\urldef\tempurl%
\url{https://doi.org/10.1609/aaai.v29i1.9744}
\showDOI{\tempurl}


\bibitem[\protect\citeauthoryear{Ismaili, Hamada, Zhang, Suzuki, and
  Yokoo}{Ismaili et~al\mbox{.}}{2019}]%
        {IsmailiHZSY19}
\bibfield{author}{\bibinfo{person}{Anisse Ismaili}, \bibinfo{person}{Naoto
  Hamada}, \bibinfo{person}{Yuzhe Zhang}, \bibinfo{person}{Takamasa Suzuki},
  {and} \bibinfo{person}{Makoto Yokoo}.} \bibinfo{year}{2019}\natexlab{}.
\newblock \showarticletitle{Weighted Matching Markets with Budget Constraints}.
\newblock \bibinfo{journal}{\emph{Journal of Artificial Intelligence Research}}
   \bibinfo{volume}{65} (\bibinfo{year}{2019}), \bibinfo{pages}{393--421}.
\newblock
\urldef\tempurl%
\url{https://doi.org/10.1613/jair.1.11582}
\showDOI{\tempurl}


\bibitem[\protect\citeauthoryear{Kamada and Kojima}{Kamada and Kojima}{2015}]%
        {kamakoji-basic}
\bibfield{author}{\bibinfo{person}{Yuichiro Kamada} {and}
  \bibinfo{person}{Fuhito Kojima}.} \bibinfo{year}{2015}\natexlab{}.
\newblock \showarticletitle{Efficient Matching under Distributional
  Constraints: Theory and Applications}.
\newblock \bibinfo{journal}{\emph{American Economic Review}}
  \bibinfo{volume}{105}, \bibinfo{number}{1} (\bibinfo{year}{2015}),
  \bibinfo{pages}{67--99}.
\newblock
\urldef\tempurl%
\url{https://doi.org/10.1257/aer.20101552}
\showDOI{\tempurl}


\bibitem[\protect\citeauthoryear{Kamada and Kojima}{Kamada and Kojima}{2017}]%
        {kamakoji-concepts}
\bibfield{author}{\bibinfo{person}{Yuichiro Kamada} {and}
  \bibinfo{person}{Fuhito Kojima}.} \bibinfo{year}{2017}\natexlab{}.
\newblock \showarticletitle{Stability Concepts in Matching under Distributional
  Constraints}.
\newblock \bibinfo{journal}{\emph{Journal of Economic Theory}}
  \bibinfo{volume}{168} (\bibinfo{year}{2017}), \bibinfo{pages}{107--142}.
\newblock
\urldef\tempurl%
\url{https://doi.org/10.1016/j.jet.2016.12.006}
\showDOI{\tempurl}


\bibitem[\protect\citeauthoryear{Kawase and Iwasaki}{Kawase and
  Iwasaki}{2017}]%
        {kawase2017near}
\bibfield{author}{\bibinfo{person}{Yasushi Kawase} {and}
  \bibinfo{person}{Atsushi Iwasaki}.} \bibinfo{year}{2017}\natexlab{}.
\newblock \showarticletitle{Near-Feasible Stable Matchings with Budget
  Constraints}. In \bibinfo{booktitle}{\emph{Proceedings of the 26th
  International Joint Conference on Artificial Intelligence (IJCAI-2017)}}.
  \bibinfo{pages}{242--248}.
\newblock
\urldef\tempurl%
\url{https://doi.org/10.24963/ijcai.2017/35}
\showDOI{\tempurl}


\bibitem[\protect\citeauthoryear{Kimura, guu Liu, Sun, Yahiro, and
  Yokoo}{Kimura et~al\mbox{.}}{2023}]%
        {Kimura23}
\bibfield{author}{\bibinfo{person}{Kei Kimura}, \bibinfo{person}{Kwei guu Liu},
  \bibinfo{person}{Zhaohong Sun}, \bibinfo{person}{Kentaro Yahiro}, {and}
  \bibinfo{person}{Makoto Yokoo}.} \bibinfo{year}{2023}\natexlab{}.
\newblock \showarticletitle{Multi-Stage Generalized Deferred Acceptance
  Mechanism: Strategyproof Mechanism for Handling General Hereditary
  Constraints}.
\newblock \bibinfo{journal}{\emph{arXiv: 2309.10968}} (\bibinfo{year}{2023}).
\newblock


\bibitem[\protect\citeauthoryear{Kojima, Tamura, and Yokoo}{Kojima
  et~al\mbox{.}}{2018}]%
        {kty:2018}
\bibfield{author}{\bibinfo{person}{Fuhito Kojima}, \bibinfo{person}{Akihisa
  Tamura}, {and} \bibinfo{person}{Makoto Yokoo}.}
  \bibinfo{year}{2018}\natexlab{}.
\newblock \showarticletitle{Designing matching mechanisms under constraints: An
  approach from discrete convex analysis}.
\newblock \bibinfo{journal}{\emph{Journal of Economic Theory}}
  \bibinfo{volume}{176} (\bibinfo{year}{2018}), \bibinfo{pages}{803--833}.
\newblock
\urldef\tempurl%
\url{https://doi.org/10.1016/j.jet.2018.05.004}
\showDOI{\tempurl}


\bibitem[\protect\citeauthoryear{Kurata, Hamada, Iwasaki, and Yokoo}{Kurata
  et~al\mbox{.}}{2017}]%
        {kurata:jair2017}
\bibfield{author}{\bibinfo{person}{Ryoji Kurata}, \bibinfo{person}{Naoto
  Hamada}, \bibinfo{person}{Atsushi Iwasaki}, {and} \bibinfo{person}{Makoto
  Yokoo}.} \bibinfo{year}{2017}\natexlab{}.
\newblock \showarticletitle{Controlled School Choice with Soft Bounds and
  Overlapping Types}.
\newblock \bibinfo{journal}{\emph{Journal of Artificial Intelligence Research}}
   \bibinfo{volume}{58} (\bibinfo{year}{2017}), \bibinfo{pages}{153--184}.
\newblock
\urldef\tempurl%
\url{https://doi.org/10.1613/jair.5297}
\showDOI{\tempurl}


\bibitem[\protect\citeauthoryear{Liu, Yahiro, and Yokoo}{Liu
  et~al\mbox{.}}{2023}]%
        {liu:spr2023}
\bibfield{author}{\bibinfo{person}{Kwei-guu Liu}, \bibinfo{person}{Kentaro
  Yahiro}, {and} \bibinfo{person}{Makoto Yokoo}.}
  \bibinfo{year}{2023}\natexlab{}.
\newblock \showarticletitle{Strategyproof Mechanism for Two-Sided Matching with
  Resource Allocation}.
\newblock \bibinfo{journal}{\emph{Artificial Intelligence}}
  \bibinfo{volume}{316} (\bibinfo{year}{2023}), \bibinfo{pages}{103855}.
\newblock
\urldef\tempurl%
\url{https://doi.org/10.1016/j.artint.2023.103855}
\showDOI{\tempurl}


\bibitem[\protect\citeauthoryear{Lu and Boutilier}{Lu and Boutilier}{2014}]%
        {lu:jmlr:2014}
\bibfield{author}{\bibinfo{person}{Tyler Lu} {and} \bibinfo{person}{Craig
  Boutilier}.} \bibinfo{year}{2014}\natexlab{}.
\newblock \showarticletitle{Effective Sampling and Learning for Mallows Models
  with Pairwise-Preference Data}.
\newblock \bibinfo{journal}{\emph{Journal of Machine Learning Research}}
  \bibinfo{volume}{15} (\bibinfo{year}{2014}), \bibinfo{pages}{3963--4009}.
\newblock


\bibitem[\protect\citeauthoryear{Mallows}{Mallows}{1957}]%
        {mallows1957non}
\bibfield{author}{\bibinfo{person}{Colin~L Mallows}.}
  \bibinfo{year}{1957}\natexlab{}.
\newblock \showarticletitle{{Non-null ranking models. I}}.
\newblock \bibinfo{journal}{\emph{Biometrika}} \bibinfo{volume}{44},
  \bibinfo{number}{1-2} (\bibinfo{year}{1957}), \bibinfo{pages}{114--130}.
\newblock
\urldef\tempurl%
\url{https://doi.org/10.2307/2333244}
\showDOI{\tempurl}


\bibitem[\protect\citeauthoryear{Murota}{Murota}{2016}]%
        {murota:dca:2016}
\bibfield{author}{\bibinfo{person}{Kazuo Murota}.}
  \bibinfo{year}{2016}\natexlab{}.
\newblock \showarticletitle{Discrete Convex Analysis: A Tool for Economics and
  Game Theory}.
\newblock \bibinfo{journal}{\emph{Journal of Mechanism and Institution Design}}
   \bibinfo{volume}{1} (\bibinfo{year}{2016}), \bibinfo{pages}{151--273}.
\newblock
\urldef\tempurl%
\url{https://doi.org/10.22574/jmid.2016.12.005}
\showDOI{\tempurl}


\bibitem[\protect\citeauthoryear{Murota and Shioura}{Murota and
  Shioura}{1999}]%
        {MS:dca:1999}
\bibfield{author}{\bibinfo{person}{Kazuo Murota} {and}
  \bibinfo{person}{Akiyoshi Shioura}.} \bibinfo{year}{1999}\natexlab{}.
\newblock \showarticletitle{{M-convex} Function on Generalized Polymatroid}.
\newblock \bibinfo{journal}{\emph{Mathematics of Operations Research}}
  \bibinfo{volume}{24}, \bibinfo{number}{1} (\bibinfo{year}{1999}),
  \bibinfo{pages}{95--105}.
\newblock
\urldef\tempurl%
\url{https://doi.org/10.1287/moor.24.1.95}
\showDOI{\tempurl}


\bibitem[\protect\citeauthoryear{Roth}{Roth}{1982}]%
        {roth1982economics}
\bibfield{author}{\bibinfo{person}{Alvin~E. Roth}.}
  \bibinfo{year}{1982}\natexlab{}.
\newblock \showarticletitle{The Economics of Matching: Stability and
  Incentives}.
\newblock \bibinfo{journal}{\emph{Mathematics of Operations Research}}
  \bibinfo{volume}{7}, \bibinfo{number}{4} (\bibinfo{year}{1982}),
  \bibinfo{pages}{617--628}.
\newblock
\urldef\tempurl%
\url{https://doi.org/10.1287/moor.7.4.617}
\showDOI{\tempurl}


\bibitem[\protect\citeauthoryear{Roth}{Roth}{1985}]%
        {roth1985}
\bibfield{author}{\bibinfo{person}{Alvin~E. Roth}.}
  \bibinfo{year}{1985}\natexlab{}.
\newblock \showarticletitle{The College Admissions Problem is not Equivalent to
  the Marriage Problem}.
\newblock \bibinfo{journal}{\emph{Journal of Economic Theory}}
  \bibinfo{volume}{36}, \bibinfo{number}{2} (\bibinfo{year}{1985}),
  \bibinfo{pages}{277--288}.
\newblock
\urldef\tempurl%
\url{https://doi.org/10.1016/0022-0531(85)90106-1}
\showDOI{\tempurl}


\bibitem[\protect\citeauthoryear{Roth and Sotomayor}{Roth and
  Sotomayor}{1990}]%
        {Roth:CUP:1990}
\bibfield{author}{\bibinfo{person}{Alvin~E. Roth} {and}
  \bibinfo{person}{Marilda A.~Oliveira Sotomayor}.}
  \bibinfo{year}{1990}\natexlab{}.
\newblock \bibinfo{booktitle}{\emph{Two-Sided Matching: A Study in
  Game-Theoretic Modeling and Analysis (Econometric Society Monographs)}}.
\newblock \bibinfo{publisher}{Cambridge University Press}.
\newblock
\urldef\tempurl%
\url{https://doi.org/10.1017/CCOL052139015X}
\showDOI{\tempurl}


\bibitem[\protect\citeauthoryear{Suzuki, Tamura, Yahiro, Yokoo, and
  Zhang}{Suzuki et~al\mbox{.}}{2023}]%
        {suzuki2022strategyproof}
\bibfield{author}{\bibinfo{person}{Takamasa Suzuki}, \bibinfo{person}{Akihisa
  Tamura}, \bibinfo{person}{Kentaro Yahiro}, \bibinfo{person}{Makoto Yokoo},
  {and} \bibinfo{person}{Yuzhe Zhang}.} \bibinfo{year}{2023}\natexlab{}.
\newblock \showarticletitle{Strategyproof Allocation Mechanisms with Endowments
  and {M}-convex Distributional Constraints}.
\newblock \bibinfo{journal}{\emph{Artificial Intelligence}}
  \bibinfo{volume}{315} (\bibinfo{year}{2023}), \bibinfo{pages}{103825}.
\newblock
\urldef\tempurl%
\url{https://doi.org/10.1016/j.artint.2022.103825}
\showDOI{\tempurl}


\bibitem[\protect\citeauthoryear{Yahiro, Zhang, Barrot, and Yokoo}{Yahiro
  et~al\mbox{.}}{2020}]%
        {Yahiro18}
\bibfield{author}{\bibinfo{person}{Kentaro Yahiro}, \bibinfo{person}{Yuzhe
  Zhang}, \bibinfo{person}{Nathana{\"{e}}l Barrot}, {and}
  \bibinfo{person}{Makoto Yokoo}.} \bibinfo{year}{2020}\natexlab{}.
\newblock \showarticletitle{Strategyproof and Fair Matching Mechanism for Ratio
  Constraints}.
\newblock \bibinfo{journal}{\emph{Autonomous Agents and Multi-Agent Systems}}
  \bibinfo{volume}{34} (\bibinfo{year}{2020}), \bibinfo{pages}{1--29}.
\newblock
\urldef\tempurl%
\url{https://doi.org/10.1007/s10458-020-09448-9}
\showDOI{\tempurl}


\end{thebibliography}


\end{document}